\documentclass[11pt]{article}
\usepackage[utf8]{inputenc}
\usepackage[english]{babel}
\usepackage{graphicx} 
\usepackage{amsmath}
\usepackage{amssymb}
\usepackage{graphicx}
\usepackage{float}
\usepackage{caption}
\usepackage{amsthm}
\usepackage{xcolor}
\usepackage{tabu}
\usepackage{fullpage}
\usepackage{hyperref}
\hypersetup{
    colorlinks,
    linkcolor={red!50!black},
    citecolor={blue!50!black},
    urlcolor={blue!80!black}
}
\usepackage{cleveref}
\usepackage{xspace}
\usepackage{thm-restate}

\iffalse
\renewbibmacro*{doi+eprint+url}{
    \printfield{doi}
    \newunit\newblock
    \iftoggle{bbx:eprint}{
        \usebibmacro{eprint}
    }{}
    \newunit\newblock
    \iffieldundef{doi}{
        \usebibmacro{url+urldate}}
        {}
    }
\fi

\theoremstyle{plain}
\newtheorem{theorem}{Theorem}
\newtheorem{lemma}[theorem]{Lemma}

\newtheorem{corollary}[theorem]{Corollary}

\theoremstyle{definition}

\title{Probabilistically checkable proofs for the Existential Theory of the Reals}

\author{Jack Stade}
\date{May 2026}

\begin{document}

\maketitle

\begin{abstract}
We prove a PCP theorem for the existential theory of the reals, showing that MAX-ETR-INV is $\exists\mathbb{R}$-hard to approximate to within some constant factor.

The existential theory of the reals (ETR) is a decision problem asking if there exists a set of real-valued variables satisfying some constraints involving polynomials and inequalities, and $\exists\mathbb{R}$ is the complexity class of problems polynomial-time reducible to ETR. Many important geometric problems are known to be $\exists\mathbb{R}$-complete. 

$\exists\mathbb{R}$-hardness results frequently work by a reduction from the $\exists\mathbb{R}$-complete problem ETR-INV, which asks if there is a an assignment of real variables each in the interval $[\frac12, 2]$ satisfying some constraints of form $x=1$, $xy=1$ and $x+y=z$.

MAX-ETR-INV is a related optimization problem that asks, given a set of constraints of form $x=1$, $xy=1$, and $x+y=z$, for a feasible (that is, satisfiable with variables in $[\frac12, 2]$) subset of those constraints of the largest possible size. We show that there is some constant $\epsilon>0$ such that it is $\exists\mathbb{R}$-hard to approximate MAX-ETR-INV with a better-than-$(1-\epsilon)$ factor. This means that even a non-deterministic polynomial-time algorithm can't approximate MAX-ETR-INV with such a factor factor unless $\exists\mathbb{R}=\text{NP}$.

We also give a polynomial-time $8$-factor approximation algorithm and a non-deterministic-polynomial-time $2$-factor approximation algorithm for MAX-ETR-INV. 
\end{abstract}

\tableofcontents

\section{Introduction}

\subsection{Background and related work}

\paragraph{The existential theory of the reals.}

The existential theory of the reals (ETR) is the decision problem for sentences of the form:

\[\exists x_1,\dots, x_n:\Phi(x_1, \dots, x_n)\]

\noindent where the $x_i$ are real variables $\Phi$ is a formula involving rational constants, arithmetic operations ($+$, $-$, $\cdot$), (in)equalities ($=$, $\le$, $<$), and logical connectives ($\wedge$, $\vee$, $\neg$). So an instance of ETR consists of constraints involving polynomials and inequalities, and asks for an assignment of real variables satisfying these constraints.

The complexity class $\exists\mathbb{R}$ consists of problems that are polynomial-time reducible to ETR. It is straightforward to see that ETR is at least NP-hard, but much less easy to see that it is even decidable. In 1988, Canny \cite{Canny1988} showed that ETR can be decided in PSPACE. So $\text{NP}\subseteq \exists\mathbb{R}\subseteq \text{PSPACE}$. It is unknown whether either inclusion is strict. Recently, improved complexity upper bounds have been claimed \cite{OpenAI2026,Meiburg2026}, putting ETR in the counting hierarchy.

The class $\exists\mathbb{R}$ has become relevant in the study of the complexity of geometric problems, with many important problems being show to $\exists \mathbb{R}$-complete. Notable examples come from art-galleries \cite{AAM2021,Stade2025}, packing \cite{AMS2020}, and training neural networks \cite{BHJMW2023,AKM2021}. A recent survey \cite{ERCompendium} lists roughly 150 $\exists\mathbb{R}$-complete problems. 

Note that a satisfying assignment for an instance of ETR can require irrational coordinates. For example, the instance $\exists x:x\cdot x=2$ has solutions with $x=\pm \sqrt{2}$. So we generally don't expect algorithms for ETR (or any $\exists\mathbb{R}$-complete problem) to output a satisfying assignment, only to determine if one exists.

\paragraph{ETR-INV.}

Many $\exists\mathbb{R}$-hardness results reduce from a particularly simple $\exists\mathbb{R}$-complete problem called ETR-INV. This is the decision problem for sentences of form:

\[\exists x_1\in \left[\tfrac12, 2\right], \dots, x_n\in \left[\tfrac12, 2\right]:\Phi(x_1, \dots, x_n)\]

\noindent where $\Phi$ is a conjunction of clauses of form $x_i=1$, $x_ix_j=1$, or $x_i+x_j=x_k$. Abrahamsen, Adamaszek, and Miltzow \cite{AAM2021} showed that ETR-INV is $\exists\mathbb{R}$-complete.

\paragraph{Hardness of approximation.}

The complexity of optimization problems is often studied in terms of a related decision problem, which asks if the optimum is better than some threshold. This models the complexity of obtaining an exactly optimal solution.

In order to study the complexity of obtaining an approximate solution, we can instead look at a related ``gap'' problem. If $A$ is a maximization problem with objective taking values in $[0, 1]$, then for $0\le a<b\le 1$, the promise problem Gap-$A_a^b$ asks, given an instance of $A$, whether the optimal objective value at least $b$. We are promised that the optimum is either $\ge b$ or $\le a$, so an algorithm for Gap-$A_a^b$ need not give a correct answer if the optimum is in $(a, b)$. 

To show hardness of approximation for $A$, we can reduce a known hard problem to Gap-$A$. A reduction to to Gap-$A_{a}^b$ must guarantee that the instance produced does not have optimum in $(a, b)$.

If there are some $a$ and $b$ such that every problem in NP reduces in polynomial time to Gap-$A_a^b$, then we say that $A$ is NP-hard to approximate to within a factor of $\frac{a}{b}$. Indeed, a better-than-$\frac{a}{b}$-factor approximation algorithm for $A$ is sufficient to distinguish between instances with optimum $\le a$ and instances with optimum $\ge b$. So a reduction from 3SAT (say) to Gap-$A_a^b$, combined with such an approximation algorithm would yield a polynomial time algorithm for 3SAT.

In a few cases, existing NP-hardness constructions can be used to obtain NP-hardness of approximation. For example, the fact that $3$-colorability of graphs is NP-hard implies that graph coloring is NP-hard to approximate to within a $\frac34$ factor. But typically, hardness of approximation is not this straightforward to show.

A common strategy to show NP-hardness is to reduce from 3SAT. So in order to show hardness of approximation, it is useful to consider the optimization version of 3SAT. The problem MAX-3SAT is an optimization problem that asks, given a 3-CNF formula, for an assignment of the variables that satisfies the largest possible fraction of the clauses. The \emph{PCP theorem} implies that MAX-3SAT is NP-hard to approximate to within some constant factor:

\begin{theorem}(PCP theorem, hardness of approximation version)\label{thm:pcptheorem}
There is some constant $c<1$ so that Gap-MAX-3$\text{SAT}_c^1$ is NP-hard. 
\end{theorem}

The theorem was first proved by Arora, Lund, Motwani, Sudan, and Szegedy \cite{ALMSS98}, building on a long line of previous work (including \cite{BFL90,BFLS91,FGLSS91,FGLSS96}).

For the formulation above, the best-known value of $c$ is due to Håstad (\cite{Hastad01}, Theorem 6.5), who proved that Gap-MAX-3$\text{SAT}_c^1$ is NP-hard for each $1>c>\frac{7}{8}$. Karloff and Zwick \cite{KZ1997} gave a randomized polynomial-time algorithm that decides Gap-MAX-3$\text{SAT}_c^1$ for any $c<\frac{7}{8}$, so this is essentially tight unless NP$\subseteq$BPP. 

Version of \Cref{thm:pcptheorem} have been used to show or improve hardness of approximation bounds for optimization versions of many classical NP-complete problems. Notable examples include maximum independent set, max cut \cite{PY1991}, set-cover \cite{RS97}, max-clique \cite{FGLSS91,FGLSS96,Zuckerman07}, and graph coloring \cite{Zuckerman07}.

\paragraph{Probabilistically checkable proofs.}

The initialization PCP comes from the theory of interactive proof systems, and stands for \emph{probabilistically checkable proofs}. We imagine a \emph{prover} (with unbounded computational resources) that wants to convince a \emph{verifier} (with polynomially-bounded resources) of the truth of some statement. The prover sends the verifier a ``proof'', which is some information that the verifier uses to help check the claim.

The PCP theorem says that a randomized protocol with polynomial proof length where the verifier only needs to read a constant number of bits of the proof can be (up to some small false-positive probability) as powerful as any protocol with polynomial-length proofs where the verifier can read the entire proof. Note that the statement being proved is agreed upon before hand, and isn't considered to be part of the proof. 

Here the bits that the verifier chooses to read depend on the outcomes of some random events. However, they are chosen \emph{non-adaptively}, that is, the verifier must choose all the bits to read before having read any of the proof. They cannot use the proof to inform decisions about which bits to read.

To see that this claim is equivalent to \Cref{thm:pcptheorem}, suppose that the prover wants to convince the verifier that an instance of SAT is satisfiable. Both the prover and verifier use the reduction implied by \Cref{thm:pcptheorem} to construct a Gap-MAX-$\text{3SAT}_c^1$ instance. Since the original SAT instance was satisfiable, the Gap-MAX-3$\text{SAT}_c^1$ will be also, and the prover can send the satisfying assignment. 

The verifier then chooses $\mathcal{O}(\frac{1}{1-c})$ clauses independently at random, and checks that the assignment sent satisfies those clauses. If the instance is not satisfiable, (that is, if the optimum is $<1$), then \Cref{thm:pcptheorem} guarantees that any assignment fails to satisfy a $1-c$ fraction of clauses, and so there is a high probability that the verifier would check one of these clauses.

So if all the clauses checked by the verifier are satisfied, then the verifier can be highly confident that the original SAT instance was satisfiable. We say that the proof sent by the prover is \emph{probabilistically checkable}. 

\paragraph{Locally testable, decodeable, and correctable codes.}

Known proofs of the PCP theorem frequently make use of error-correcting codes with certain ``local'' properties.

A code $C\subseteq \Sigma^m$ over an alphabet $\Sigma$ is called $q$-local $\delta$-testable if there is a randomized process that chooses (non-adaptively) $q$ entries of a string $x\in \Sigma^m$ and accepts or rejects based on those queries. The test should satisfy:

\begin{itemize}
    \item (Completeness) if $x\in C$, then the test accepts with probability $1$
    \item (Soundness) if the $x\in \Sigma^m$ is such that the test accepts with high probability, then there is some $y\in C$ such that $x$ and $y$ differ on at most a $\delta$ fraction of entries (we say that $x$ and $y$ are $\delta$-close).
\end{itemize}

A code $C\subseteq \Sigma^m$ is $(q, \delta, \epsilon)$-locally correctable if there is a randomized process that given an index $i\in [m]$, reads $q$ entries of a string $x\in \Sigma^m$ and guesses an element of $\Sigma$. If $x$ is $\delta$-close to some $y\in C$, then the process should correctly determine the $i$th entry of $y$ with probability at least $1-\epsilon$.

A code $f:\Sigma^n\rightarrow \Sigma^m$ is $(q, \delta, \epsilon)$-locally decodeable if there is a randomized process that given an index $i\in [n]$, reads $q$ entries of a string $x\in \Sigma^m$ and guesses an element of $\Sigma$. If $x$ is $\delta$-close to $f(y)$ for some $y\in \Sigma^n$, then the process should correctly determine the $i$th entry of $y$ with probability at least $1-\epsilon$.

A variety of constant-query locally testable, decodeable, and/or correctable codes are known over finite alphabets (see e.g. \cite{Goldreich2010,DELLM22,Yekhanin2012,BSS2020}). Over an infinite alphabet ($\mathbb{R}$, say), it is possible to create codes with unreasonably good properties, for example using an injective map $\mathbb{R}^n\rightarrow \mathbb{R}$. To rule out these examples, we can restrict to linear codes, that is codes given by a linear subspace of $C\subseteq \mathbb{R}^m$. 

Constant-query locally-decodeable linear codes are known over $\mathbb{R}$, and recently Dvir \cite{Dvir2026} has claimed a construction of constant-query locally-correctable linear codes over $\mathbb{R}$.

\subsection{Our results}

Let MAX-ETR-INV be the optimization version of ETR-INV. So an instance of MAX-ETR-INV consists of some constraints of form $x=1$, $xy=1$, or $x+y=z$, and asks to determine the largest fraction of clauses that can be satisfied simultaneously by an assignment of variables in $[\frac12, 2]$. Our main result is the following:

\begin{theorem}\label{thm:main}
There is some constant $c<1$ such that Gap-MAX-ETR-INV$_c^1$ is $\exists\mathbb{R}$-hard.
\end{theorem}

This says that even a non-deterministic polynomial-time algorithm can't approximate MAX-ETR-INV better than a $c$ factor, unless $\text{NP}=\exists\mathbb{R}$. 

In order to prove \Cref{thm:main}, we first show hardness of approximation for a somewhat more general problem that we call ETR-$\mathcal{C}(q)$. For each $q>0$, the problem ETR-$\mathcal{C}(q)$ asks if it is possible to assign in $[\frac12, 2]$ to some variables so that some constraints of form $x=y$, $x=(1+q)y$, $x+y=z$, $x+y=z+w$, $xy=z$, and $xy=zw$ are satisfied. The problem MAX-ETR-$\mathcal{C}(q)$ asks for the largest number of clauses that can be satisfied by an assignment. 

We first prove that, for sufficiently small $q$, it is $\exists\mathbb{R}$-hard to approximate MAX-ETR-$\mathcal{C}(q)$ to within some constant factor. In \Cref{sec:formulaconversion}, we then use more-or-less standard algebraic manipulations to show that any constraint appearing in MAX-ETR-$\mathcal{C}(q)$ can be simulated with $\mathcal{O}(\log(\frac1q))$ constraints of form $x=1$, $xy=1$, or $x+y=z$. By choosing some small fixed $q$, this will prove \Cref{thm:main}.

Our proof that MAX-ETR-$\mathcal{C}(q)$ is $\exists\mathbb{R}$-hard to approximate is based on a proof of the PCP theorem by Dinur \cite{Dinur2007}, and some simplifications of that proof by Radhakrishnan \cite{Radha2006}. 

We discuss Dinur's proof (and the improvements by Radhakrishnan) in \Cref{sec:dinur}. The proof proceeds in three steps, and in \Cref{sec:steps12}, we notice that the first two steps can already be applied in our setting.

The third step is the construction of an \emph{assignment tester}. An assignment tester is like a PCP protocol for satisfiability of a constraint satisfaction problem, but the proof also contains an assignment of the variables in the original instance. In addition to verifying that the instance is satisfiable, the verifier also needs to check that the assignment provided is close (in some sense) to a satisfying assignment.

Roughly speaking, the first two steps in Dinur's proof show that in order to prove a PCP theorem for a constraint satisfaction problem, it is enough to construct an assignment tester for that problem, where the proof can have unbounded length and the verifier has unbounded computational resources (but is still only allowed a constant number of queries to the proof).

Existing assignment-tester constructions work only for discrete alphabets, and constructing one for a problem like ETR-$\mathcal{C}(q)$ requires some new ideas. In \Cref{sec:step3}, we fully state our main technical result, which is a type of assignment tester for ETR-$\mathcal{C}(q)$. We then explain how to obtain inapproximability of Max-ETR-$\mathcal{C}(q)$ from this result. The construction of our assignment tester occupies \Cref{sec:code,sec:constraints,sec:finalproof}. 

\paragraph{Codes over infinite alphabets.}

Our assignment tester is based on a novel way of testing and decoding linear codes over $\mathbb{R}$. For an abelian group $G$ and a subset $S$ containing $m$ points in $\mathbb{Z}^m$, there is a linear code $G^n\rightarrow G^m$ that sends $\alpha\in G^n$ to the map $S\rightarrow G$ sending $x\in S$ to $\sum_ix_i\alpha_i$.

There has already been some work studying these codes when $G$ is the additive group of the reals. For example, Amireddy, Behera, Paraashar, Srinivasan and Sudan \cite{ABPSS2024} showed that the codes with $S=\{0, 1\}^n$ is locally correctable with $\mathcal{O}(\log(n))$ queries for any abelian group $G$. Some known lower bounds \cite{BDYW2011,DSW2014,DSW2017,BSS2020,ABPSS2025} illustrate that locally correction is more difficult when $G$ has elements of infinite order. 

Setting $S=\{0, \dots, k\}^n$ (for some parameter $k$), we give a randomized tester and decoder for the corresponding code $G^n\rightarrow G^{m}$, where $m=(k+1)^n$. The tester queries a string in $G^m$ in $4$ places and either accepts or rejects the string. The decoder takes as input an integer vector $d\in \mathbb{Z}^n$ with $\Vert d\Vert_1$ sufficiently small, makes two random queries to a string $G^m$ and outputs a value in $G$.

If the $\beta\in G^m$ is the encoding of some $\alpha\in G^n$, then tester will accept $\beta$ with probability $1$ and the decoder will always output $\sum_id_i\alpha_i$. We show that if $\beta\in G^m$ causes the tester to accept with high probability, then there is some $\alpha\in G^n$ such that, for any decoding parameter $d$, the decoder is likely to correctly evaluate $\sum_id_i\alpha_i$. 

These codes are discussed in more detail in \Cref{sec:code}. In \Cref{sec:constraints}, we show how to use our codes to test $G$-linear constraints. Then we show how to use two instances of our code, one with $G$ being the additive group of the reals and one with $G$ being the multiplicative group of the reals, to test complicated constraints defined by polynomials. In \Cref{sec:finalproof}, use this test to construct our assignment tester, which will complete the proof of \Cref{thm:main}.

\paragraph{Interactive proofs with real RAM.}

The main motivation for studying MAX-ETR-INV comes from potential applications to $\exists\mathbb{R}$-hardness of approximation. However, we can also give a nice interpretation of our results in terms of interactive proof systems.

We should imagine an interactive proof system where the proof sent by the prover consists of real numbers (rather than binary bits), and the verifier has access to real RAM. So the verifier can add, subtract, multiply, divide and compare real variables at unit cost per operation. 

With polynomial many operations, a verifier can extract $n$ bits of information from a single real variable, which perhaps goes against the spirit of probabilistically checkable proofs. So it is natural to restrict the verifier further, by limiting the number of arithmetic operations that can be performed on real variables.

This is exactly what \Cref{thm:main} does. The theorem implies that ETR has polynomial-length, constant-query probabilistically-checkable proofs in this model, even if the verifier is only allowed to perform a constant number of arithmetic operations on real RAM (and even if the only operations allowed are checking constraints of form $x=1$, $xy=1$, $x+y=z$, or $x\in [\frac12, 2]$).

\paragraph{Complexity of solutions.}

A satisfying assignment to an instance of ETR-INV can require irrational coordinates, but if a solution exists, then one exists where the coordinates are real algebraic numbers (that is, roots of integer-coefficient univariate polynomials, this is implied by e.g. the algorithm from \cite{Canny1988}).

It is possible to perform arithmetic operations on real algebraic numbers. For example, Mishra and Pedersen \cite{MP90} give a nearly canonical way of representing a real algebraic as a root of a polynomial with integer coefficients, and show that basic arithmetic operations can be performed in time polynomial in the size of this representation. 

However, a satisfiable instance of ETR-INV might require coordinates that have exponential complexity in this representation. Our techniques can be used to show that this extends to approximating ETR-INV. Recall that a real algebraic number $\alpha$ generates a field extension $\mathbb{Q}(\alpha)$ of $\mathbb{Q}$ of finite degree, where the degree of the extension is the degree of the lowest-degree integer-coefficient polynomial of which $\alpha$ is a root.

\begin{theorem}\label{thm:solutioncomplexity}
There is some constant $0<c<1$ such that, for any $k\in \mathbb{N}$, there is a satisfiable instance of ETR-INV with $n\ge 2$ variables such that any assignment satisfying at least a $c$ fraction of constraints requires at least one variable to be set to a value that generates a field extension of degree more than $n^k$ over $\mathbb{Q}$.
\end{theorem}

The representation from \cite{MP90} represents a real algebraic number using a polynomial of which it is a root, so we need at least $d$ bits to represent a number that generates a field extension of degree $d$. So \Cref{thm:solutioncomplexity} implies that we need super-polynomially many bits to represent even an approximate solution to ETR-INV. We give a proof of \Cref{thm:solutioncomplexity} in \Cref{sec:solutioncomplexity}. 

\paragraph{Approximation algorithms for ETR-INV.}

In \Cref{sec:algorithms}, we give two constant-factor approximation algorithms for ETR-INV. The first algorithm is a polynomial-time $8$-factor approximation algorithm. This is based on a fairly standard probabilistic argument.

Using the additional power of NP, we get a better approximation algorithm with an approximation factor of $2$. The idea is to try to satisfy either all the constraints of form $x=1$ and $xy=1$ or as many of the constraints of form $x+y=z$ as possible. We need the power of NP because this requires solving the NP-hard (and hard to approximate \cite{EV95}) maximum-feasible-subsystem problem. 

Both algorithms \Cref{sec:algorithms} can output assignments achieving the fraction of satisfied constraints that they report. The assignments that they output are rational numbers with polynomially many bits. 

\section{Dinur's PCP construction}\label{sec:dinur}

An \emph{constraint graph} is a graph $G$ equipped with an alphabet $\Sigma$, and, for each edge, a constraint modeled as a subset of $\Sigma^2$. Throughout, we will consider constraints graphs that can have loops and multiple edges. An assignment of a constraint graph is a map $\mathcal{A}:V(G)\rightarrow \Sigma$. We say that an edge $e=(u, v)$ is satisfied by an assignment if $(\mathcal{A}(u), \mathcal{A}(v))$ is in the corresponding subset of $\Sigma^2$. 

Given a constraint graph $G$, let $\text{UNSAT}(G)$ be the smallest fraction of edges that are not satisfied over all assignments. A graph is satisfiable if $\text{UNSAT}(G)=0$. 

Dinur proves that, given a constraint graph $G$ with finite alphabet ($\Sigma=\{0, 1\}^3$, say), there is a polynomial-time construction of a constraint graph $G'$ with alphabet $\{0, 1\}^3$ such that, if $\text{UNSAT}(G)=0$ then $\text{UNSAT}(G')=0$, but if $\text{UNSAT}(G)>0$ then $\text{UNSAT}(G')\ge 2\cdot\text{UNSAT}(G)$. The construction proceeds in $3$ steps:

\begin{enumerate}
    \item Preprocess $G$ so that it is $d$-regular (for some global constant $d$) and satisfies certain expansion properties. Let $G_1$ be the graph obtained. This satisfies $\text{UNSAT}(G_1)\ge \eta_1\cdot \text{UNSAT}(G)$ for some universal constant $\eta_1$. 
    \item Let $G_2=G_1^t$ for some constant $t$. This is a new graph satisfying\footnote{In Dinur's original construction, the bound was $\text{UNSAT}(G_2)\ge \eta_2\sqrt{t}\cdot \text{min}(\text{UNSAT}(G_1), \frac1t)$. We use an improvement due to Radhakrishnan \cite{Radha2006} that achieves the stated bound.} $\text{UNSAT}(G_2)\ge \eta_2t\cdot \text{min}(\text{UNSAT}(G_1), \frac1t)$ for some global constant $\eta_2$. However, $G_2$ now has alphabet $\{0, 1\}^{3d^{t+1}}$. 
    \item Compose $G_2$ with an assignment tester to get a graph $G'$ with alphabet $\{0, 1\}^{3}$. This has $\text{UNSAT}(G')\ge \eta_3\cdot \text{UNSAT}(G_2)$ for some global constant $\eta_3$. 
\end{enumerate}

Each step preserves consistency, that is, if $\text{UNSAT}(G)=0$ then $\text{UNSAT}(G')=0$. The constants $\mu_1$, $\mu_2$ and $\mu_3$ don't depend on $t$, so by choosing $t=\frac{2}{\eta_1\eta_2\eta_3}$, we have $\text{UNSAT}(G')\ge 2\cdot \text{UNSAT}(G)$.

Each step also causes the number of edges in the graph to increase by at most a constant factor (depending only on $t$). So the process can be applied logarithmically many times in polynomial total time. 

Step (1) can be applied to constraint graphs with real variables and polynomial constraints with essentially no modification. 

In Dinur's original construction, the constant $\eta_2$ depends on the size of the alphabet for $G$. However, Radhakrishnan \cite{Radha2006} later developed a simplified version of step (2). Crucially, in his proof the constant does not depend on the alphabet size. This version can be easily adapted to our setting. 

In contrast, step (3) requires several new ideas. Our main contribution is to construct an assignment tester for families of continuous constraints. 

\subsection{Steps (1) and (2)}\label{sec:steps12}

We now state formal versions of steps (1) and (2) that work in the continuous setting, being careful to keep track of the types of constraints needed.

We first give a couple of definitions to help states the results. If $\mathcal{C}$ is a constraint set on $\Sigma^{\ell}$, then we define the constraint set $\mathcal{C}_{\Sigma}^{k}$ on $\Sigma^k$ to be the set of ways to apply a constraint in $\mathcal{C}$ to $\ell$ (not necessarily distinct) entries from $\Sigma^k$. That is:

\[\mathcal{C}_{\Sigma}^k=\{\{(x_1, \dots, x_k)\in \Sigma^k|(x_{i_1}, \dots, x_{i_\ell})\in c\}|c\in\mathcal{C}, i_1, \dots, i_\ell\in [k]\}\]

For a constraint set $\mathcal{C}$ on $\Sigma$, we define $\mathcal{P}(\mathcal{C})$ to be the constraint set on $\Sigma$ given by arbitrary conjunctions of constraints in $\mathcal{C}$. That is, for each $\mathcal{S}\subseteq \mathcal{C}$, $\mathcal{P}(\mathcal{C})$ has a constraint:

\[\bigcap_{c\in\mathcal{S}}c\]

The statements of steps (1) and (2) involve the spectral expansion of a graph. For a graph $G$, the second eigenvalue $\lambda(G)$ is the second largest (in absolute value) eigenvalue of $G$'s adjacency matrix. See e.g. \cite{HLW2006} for an exposition on expander graphs. Outside the statements of steps (1) and (2), we will not need the theory of expander graphs. 

\begin{lemma}(Step (1))\label{lem:step1}
There are some universal constants $\eta_1$, $d$ and $\lambda<d$ such that, for a constraint graph $G$ with alphabet $\Sigma$ and constraints from a set $\mathcal{C}$, we can compute a $d$-regular constraint graph $G_1$ with alphabet $\Sigma$ and constraints from $\mathcal{C}\cap \{x=y, \Sigma^2\}$ such that:

\begin{itemize}
    \item $|E(G_1)|=\mathcal{O}(|E(G)|)$
    \item $\lambda(G_1)\le \lambda$
    \item if $\text{UNSAT}(G)=0$ then $\text{UNSAT}(G_1)=0$
    \item $\text{UNSAT}(G_1)\ge \eta_1\cdot \text{UNSAT}(G)$
\end{itemize}
\end{lemma}

For a proof of \Cref{lem:step1}, see Dinur (\cite{Dinur2007}, Section 4). While the result is stated for finite alphabets, the proof works for infinite alphabets without modification. Note that the constants $\eta_1$, $d$ and $\lambda$ do not depend on $\Sigma$. 

The only new types of constraints needed to construct $G_1$ are a null constraint and an equality constraint $\{(x, y)\in \Sigma^2|x=y\}$, written $x=y$ for short. Note that if $\Sigma=\mathbb{R}^4$ (say), then this represents a conjunction of $4$ constraints of form $x=y$ for $x,y\in \mathbb{R}$.

We can now state step 2:

\begin{lemma}(Step (2))\label{lem:step2}
Fix constants $d$ and $\lambda<d$. Then there is some constant $\eta_2$ such that, for each $t\in \mathbb{N}$, given a constraint $d$-regular constraint graph $G_1$ with alphabet $\Sigma$, constraint set $\mathcal{C}$ and $\lambda(G_1)\le \lambda$, we can in polynomial time construct a constraint graph $G_2$ with alphabet $\Sigma^{d^{t+1}}$ and constraint set $\mathcal{P}(\mathcal{C}_{\Sigma}^{2d^{t+1}})$ such that:

\begin{itemize}
    \item $|E(G_2)|=\mathcal{O}_t(|E(G_1)|)$
    \item if $\text{UNSAT}(G_1)=0$, then $\text{UNSAT}(G_2)=0$
    \item $\text{UNSAT}(G_2)\ge \eta_2t\cdot \text{min}(\text{UNSAT}(G_1), \frac1t)$
\end{itemize}
\end{lemma}

For a proof of \Cref{lem:step2}, see Radhakrishnan \cite{Radha2006}. Unlike the gap amplification lemma in \cite{Dinur2007}, the constant $\eta_2$ depends on only on $d$ and $\lambda$ (and not on $|\Sigma|$). So the proof works for infinite alphabets without modification.

Say $G$ had constraint set $\mathcal{C}$. Then $G_2$ has constraint set $\mathcal{P}((\mathcal{C}\cup \{x=y, \Sigma^2\})_{\Sigma}^{2d^{t+1}})$. Recall that this consists of unions of constraints from $\mathcal{C}\cup \{x=y, \Sigma^2\}$ each applied to a subset of entries with indices $[2d^{t+1}]$.

\section{Assignment testers and alphabet reduction}\label{sec:step3}

Our version of step (3) works by constructing a type of assignment tester for certain semialgebraic sets. We should first define carefully what we mean by an assignment tester.

Let $\mathcal{C}$ be a constraint set on an alphabet $\Sigma^2$ (that is, a set of subsets of $\Sigma^2$). The type of assignment tester that we will need consists of the following data:

\begin{itemize}
    \item a blow-up size $k$
    \item an output alphabet $\Gamma$
    \item an error-correcting function $f:\Sigma\rightarrow \Gamma^k$
    \item a number $m$ of auxiliary variables
    \item a randomness parameter $R$
    \item a length parameter $\ell$
    \item an error threshold $\zeta$
    \item an output constraint set $\mathcal{D}$ on $\Gamma^\ell$, and
    \item a test function $(\mathcal{Q}, \mathcal{T}): \mathcal{C}\times [R]\rightarrow [2k+m]^\ell\times \mathcal{D}$
\end{itemize}

To test a constraint $c\in \mathcal{C}$, the assignment test takes as input a value $\alpha\in \mathbb{R}^{2k+m}$, then chooses a value in $r\in [R]$ uniformly at random, queries the $\ell$ entries of $\alpha$ given by $\mathcal{Q}(c, r)$, and checks that they satisfy the constraint given by $\mathcal{T}(c, r)$. 

Let $\delta$ be the distance of the error-correcting code $f$. That is, $\delta$ is such that $f(x)$ and $f(y)$ differ on at least a $\delta$ fraction of entries for any distinct $x$ and $y$ in $\Sigma$. This test should have the following properties:

\begin{itemize}
    \item (Completeness) for each pair $(x, y)\in \Sigma^2$ satisfying $c$, there is some $\alpha$ such that the test always passes on input $(f(x), f(y), \alpha)$
    \item (Soundness) if $(\alpha, \beta, \gamma)$ causes the test to pass with probability at least $1-\zeta$, then there are $x, y\in \Sigma$ satisfying the constraint $c$ such that $\alpha$ is less than $\frac{\delta}{2}$-close to $f(x)$ and $\beta$ is less than $\frac{\delta}{2}$-close to $f(y)$
\end{itemize}

This construction is adapted from Dinur and Reingold's definition of an assignment tester \cite{DR2006}, with the most important difference being that our assignment tester includes an error-correcting function $f$ as part of the definition. This is used to ensure what Dinur and Reingold call ``robustness''. Our construction in \Cref{sec:finalproof} relies on a specific choice of error correcting function, so we include the error correcting function as part of the data of an assignment tester.

Given a constraint graph with a large alphabet, we can use \emph{composition with an assignment tester} to get a graph over a smaller alphabet. We will use the following composition lemma:

\begin{lemma}(Composition Lemma)\label{lem:composition}
Let $G_2$ be a constraint graph with alphabet $\Sigma$ and constraint set $\mathcal{C}$. Given an assignment tester for $\mathcal{C}$ with parameters as above, we can construct a new graph $G'$ with alphabet $\Gamma^\ell$ and constraint set $\mathcal{D}_{[\frac12, 2]}^{2\ell}$ such that:

\begin{itemize}
    \item if $\text{UNSAT}(G_2)=0$, then $\text{UNSAT}(G')=0$
    \item $\text{UNSAT}(G')\ge \frac{\zeta}{\ell+1}\cdot \text{UNSAT}(G)$
    \item $|E(G')|=(\ell+1)R|E(G)|$
\end{itemize}
\end{lemma}

\begin{proof}
For each vertex $u$ of $G_2$, we create $k$ vertices $u_1, \dots, u_k$ of $G'$. For each edge $e=(u, v)$ of $G_2$ we create $m$ vertices $e_1, \dots, e_m$ and $R$ vertices $t_{e, 1}, \dots, t_{e, R}$. The graph $G'$ should have alphabet $\Gamma^{\ell}$. So each vertex represents a variable in $\Gamma^{\ell}$, or equivalently a vector with $\ell$ entries representing values in $\Gamma$. For the $u, v$ and $e$ vertices, only the first entry will be used.

Fix an edge $e=(u, v)$ of $G_2$. This edge represents a constraint $c_e\in \mathcal{C}$. For each $r\in \{1, \dots, R\}$ and $i\in \{1, \dots, \ell\}$, we create an edge between $t_{e, r}$ and the $\mathcal{Q}_i(c_e, r)$-index vertex in $(u_1, \dots, u_k, v_1, \dots, v_k, e_1, \dots, e_m)$. This represents the constraint that the $i$th entry of $t_{e, r}$ equals the first entry of the other vertex at the other end of the edge. 

For each $t_{e, r}$, we add a self loop to that vertex representing the constraint $\mathcal{R}(c_e, r)$.

If $\mathcal{A}$ is a satisfying assignment of $G_2$, then we can construct a satisfying assignment $\mathcal{B}$ of $G'$. Set $\mathcal{B}(u_i)=(f_i(\mathcal{A}(u)), \dots, f_i(\mathcal{A}(u)))$. Here we are writing $f:\Sigma\rightarrow \Gamma^k$ as $(f_1, \dots, f_k)$, with $f_i:\Sigma\rightarrow \Gamma$. Note that only the first entry of $\mathcal{B}(u_i)$ is relevant for the construction, the remaining entries are just padding.

For $e=(u, v)$ and edge of $G_2$, we know that $(\mathcal{A}(u), \mathcal{A}(v))$ satisfies the constraint on that edge, so the properties of an assignment tester guarantee that there is an assignment of the auxiliary variables, which can be used to set the $\mathcal{B}(e_i)$. 

Finally, the $\ell$ entries of each $t_{e, r}$ can be uniquely defined by the equality constraints connecting it to the vertices specified by $\mathcal{Q}(c_e, r)$. The properties of an assignment tester ensure that the loop on $t_{e, r}$ is satisfied.

Conversely, suppose that we have some assignment $\mathcal{B}$ of $G'$ that satisfies a $1-\epsilon$ fraction of edges. For each vertex $u$ of $G_2$, let $\mathcal{A}(u)$ be the element $x\in \mathbb{R}^n$ that minimizes the (Hamming) distance between $f(x)$ and $(\mathcal{B}(u_1), \dots, \mathcal{B}(u_k))$. 

Let $S$ be the set of edges $e$ of $G_2$ such that, for at least a $1-\zeta$ fraction of $r\in \{1, \dots, R\}$, all the edges incident to $t_{e, r}$ are satisfied by $\mathcal{B}$. 

For an edge $e=(u, v)$ in $S$, the assignment $(\mathcal{B}(u_1), \dots, \mathcal{B}(u_k), \mathcal{B}(v_1), \dots, \mathcal{B}(v_k), \mathcal{B}(e_1), \dots, \mathcal{B}(e_m))$ causes the assignment tester to pass with probability at least $1-\zeta$, so there are some $x, y\in \mathbb{R}^n$ satisfying $c_e$ such that $(\mathcal{B}(u_1), \dots, \mathcal{B}(u_k))$ is $<\frac12\delta$-close to $f(x)$ and $(\mathcal{B}(v_1), \dots, \mathcal{B}(v_k))$ is $<\frac12\delta$-close to $f(y)$. Since $f$ has distance $\delta$, this means that $x=\mathcal{A}(u)$ and $y=\mathcal{A}(v)$, so the edge $e$ is satisfied by $\mathcal{A}$.

Each edge of $G'$ is incident to exactly one of the $t_{e, r}$ vertices, and each $t_{e, r}$ vertex is incident to exactly $\ell+1$ edges. So at least a $1-(\ell+1)\epsilon$ of the $t_{e, r}$ vertices have all incident edges satisfied. 

For each edge $e$ of $G_2$, there are $R$ vertices $t_{e, r}$. So at least a $1-\frac{(\ell+1)\epsilon}{\zeta}$ fraction of edges $e$ of $G_2$ are in $S$. So $\text{UNSAT}(G')\ge \frac{\mu_3}{(\ell+1)}\cdot \text{UNSAT}(G)$. 
\end{proof}

Similar compositions lemmas appear in e.g. \cite{DR2006,Dinur2007}.

\subsection{The constraint set $\mathcal{C}(q)$}\label{sec:cq}

For $q>0$, the constraint set $\mathcal{C}(q)$ contains the following constraints on $[\frac12, 2]^4=\{(x, y, z, w)\in [\frac12, 2]^4\}$:

\[\begin{array}{c}
x=y\\
x=(1+q)y\\
x+y=z\\
x+y=z+w\\
xy=z\\
x+y=z+w\\
-\\
\end{array}\]

The last constraint in the list is the null constraint, which is included for convienience. Our main technical result is an assignment tester for constraint sets using constraints of this form:

\begin{lemma}\label{lem:main}
For each $n$, for sufficiently small $q>0$, we can construct an assignment tester on $\mathcal{P}\left(\mathcal{C}(q)_{[\frac12, 2]}^{n}\right)$ with length $\ell=4$, error threshold $\zeta=\frac{1}{1000}$, output alphabet $\Gamma=[\frac12, 2]$, and output constraint set $\mathcal{D}=\mathcal{C}(q)_{[\frac12, 2]}^4$
\end{lemma}

Note that the value of $q$ needed depends on $n$, with larger values of $q$ requiring larger values of $n$. However, the error threshold $\zeta=\frac{1}{1000}$ does not depend on $n$. The proof of \Cref{lem:main} can be found in \Cref{sec:code,sec:constraints,sec:finalproof}. In the rest of this section, we explain how \Cref{lem:main} leads to a proof of \Cref{thm:main}. 

\begin{corollary}(Step 3)\label{cor:step3}
Fix $n$. Let $G_2$ be a constraint graph with alphabet $[\frac12, 2]^{n}$ and constraint set $\mathcal{P}\left(\mathcal{C}(q)_{[\frac12, 2]}^{2n}\right)$. If $q$ is sufficiently small, then in polynomial time we can construct a constraint graph $G_3$ with alphabet $[\frac12, 2]^4$ and constraint set $\mathcal{C}(q)_{[\frac12, 2]}^{8}$ such that:

\begin{itemize}
    \item $|E(G_3)|=\mathcal{O}_n(|E(G_2)|)$
    \item if $\text{UNSAT}(G_2)=0$ then $\text{UNSAT}(G_3)=0$
    \item $\text{UNSAT}(G_3)\ge \eta_3\cdot \text{UNSAT}(G_2)$, where $\eta_3=\frac{1}{5000}$
\end{itemize}
\end{corollary}

\begin{proof}
Straightforward application of \Cref{lem:main} and \Cref{lem:composition}.
\end{proof}

With steps (1), (2) and (3), we can prove a gap amplification lemma:

\begin{lemma}\label{lem:gapamplification}
There is some $\alpha>0$ such that, for sufficiently small $q$, given a constraint graph $G$ with alphabet $[\frac12, 2]^4$ and constraint set $\mathcal{C}(q)_{[\frac12, 2]}^8$, we can construct in polynomial time a constraint graph $G'$ with the same alphabet and constraint set, such that:

\begin{itemize}
    \item $|E(G')|=\mathcal{O}(E(G))$
    \item if $\text{UNSAT}(G)=0$ then $\text{UNSAT}(G')=0$
    \item $\text{UNSAT}(G')\ge \text{min}(\alpha, 2\cdot \text{UNSAT}(G))$
\end{itemize}
\end{lemma}

\begin{proof}
Let $t=\lceil\frac{2}{\eta_1\eta_2\eta_3}\rceil$, where $\eta_1, \eta_2$ and $\eta_3$ are the constants from \Cref{lem:step1}, \Cref{lem:step2}, and \Cref{cor:step3} respectively. Suppose $q$ is small enough that \Cref{cor:step3} holds with $n=8d^{t+1}$, where $d$ is the constant from \Cref{lem:step1}. 

Let $\alpha=\eta_2\eta_3$. Applying \Cref{lem:step1,lem:step2} to $G$, we get a graph $G_2$ with:

\[\text{UNSAT}(G_2)\ge \eta_2t\cdot\text{min}\left(\eta_1\cdot\text{UNSAT}(G), \frac1t\right)\ge \frac{1}{\eta_3}\text{min}\left(2\cdot \text{UNSAT}(G), \alpha\right)\]

$\mathcal{C}(q)$ contains the null constraint, so the constraint set of $G_2$ is 

\[\mathcal{P}\left((\mathcal{C}(q)_{[\frac12, 2]}^8\cup \{x=y, [\frac12, 2]^8\})_{[\frac12, 2]^4}^{2d^{t+1}}\right)\]. 

Since $\mathcal{C}(q)$ contains a constraint $x=y$, it is not too hard to see that this is contained in $\mathcal{P}\left(\mathcal{C}(q)_{[\frac12, 2]}^{8d^{t+1}}\right)$. 

So we can apply \Cref{cor:step3} to obtain $G'$. This has $\text{UNSAT}(G')\ge \text{min}(2\cdot \text{UNSAT}(G), \alpha)$, and the other required properties are straightforward to verify.
\end{proof}

\subsection{Inapproximability of Max-ETR-$\mathcal{C}(q)$}

For each $q>0$, we define ETR-$\mathcal{C}(q)$ to be real constraint satisfaction problem with constraints of the form in $\mathcal{C}(q)$. So an instance of ETR-$\mathcal{C}(q)$ asks for an assignment of real variables in $[\frac12, 2]$ that satisfy constraints of form $x=y$, $x=(1+q)y$, $x+y=z$, $x+y=z+w$, $xy=z$, and $xy=zw$.

We also define the optimization version MAX-ETR-$\mathcal{C}(q)$, which asks for assignment satisfying the largest possible fraction of constraints. For $0\le a<b\le 1$ there is a corresponding gap problem Gap-MAX-ETR-$\mathcal{C}(q)_{a}^b$. 

Using \Cref{lem:gapamplification}, we can can obtain a hardness of approximation result for MAX-ETR-$\mathcal{C}(q)$:

\begin{theorem}\label{thm:etrq}
For sufficiently small $q>0$, there is some $0\le c<1$ such that ETR-$\mathcal{C}(q)$ polynomial-time reduces to Gap-MAX-ETR-$\mathcal{C}(q)_{c}^1$.
\end{theorem}

\begin{proof}
An instance of ETR-$\mathcal{C}(q)$ can be represented as a constraint graph $G$ with alphabet $[\frac12, 2]^4$ and constraint set $\mathcal{C}(q)_{[\frac12, 2]}^8$. If the instance is satisfiable, then $\text{UNSAT}(G)=0$. Otherwise, $\text{UNSAT}(G)\ge \frac{1}{|E(G)|}$. 

We can then apply \Cref{lem:gapamplification} to $G$ a total of $\lceil\log_2(|E(G)|)-\log_2(\alpha)\rceil$ times to obtain a graph $G'$. Each iteration increases the size of the graph by at most a linear factor, so $|E(G')|$ is at most polynomial in the size of $G$. Each step runs in time polynomial in the size of its input, so the entire process takes polynomial time.

If $\text{UNSAT}(G)=0$, then $\text{UNSAT}(G')=0$. Otherwise, $\text{UNSAT}(G)\ge \frac{1}{|E(G)|}$, so we have $\text{UNSAT}(G')\ge \alpha$, where $\alpha$ is the constant from \Cref{lem:gapamplification}. So this gives a polynomial time reduction from ETR-$\mathcal{C}(q)$ to Gap-MAX-ETR-$\mathcal{C}(q)_{1-\alpha}^1$
\end{proof}

In \Cref{sec:formulaconversion}, we give a gap-preserving reduction from each ETR-$\mathcal{C}(q)$ to ETR-INV. Combined with \Cref{thm:etrq} and a proof that ETR-$\mathcal{C}(q)$ is $\exists\mathbb{R}$-hard, this proves \Cref{thm:main}. 

\section{The midpoint code}\label{sec:code}

In this section, we describe the \emph{midpoint code}, which our main technical contribution. For an abelian group $G$, this will be a $\mathbb{Z}$-linear map $C:G^{n}\times G\rightarrow G^m$. We view this as allowing an element of $G^n$ to be encoded many different ways, one for each element of $G$. 

For the purposes of proving \Cref{thm:main}, we will take $G$ to be either the multiplicative or additive group of the reals. But the results in this section hold for arbitrary abelian groups. 

We identify $[m]$ with a finite subset $R$ of $\mathbb{Z}^n$. The code sends $(\beta, g)\in G^n\times G$ to the map $R\rightarrow G$ sending $x\in R$ to $g+\sum_i\beta_ix_i$.

\paragraph{Testers and decoders.}

We will describe a \emph{tester} and a \emph{decoder} for our code $C$. Here we describe the precise properties that we want these to have.

The tester is a process that, given an element $\beta\in G^m$, probabilistically chooses a random (but constant size) set of indices in $[m]$, reads those entries in $\beta$, and performs some test using the values read. If $\beta$ is a codeword (that is, it is in the image of $C$), then the test should pass with probability $1$.

The decoder is a process that, given $\beta\in G^m$, tries to determine some information about any $\alpha\in G^n$ that $\beta$ encodes or is close to encoding. The decoder takes as input a small-integer vector $d\in \mathbb{Z}^n$ and guesses the value of $\sum_{i=1}^nd_i\alpha_i$. The decoder should probabilistically choose a constant-size set of indices to evaluate $\beta$ at, and guess the value of $\sum_id_i\alpha_i$ based on those values. If $\beta=C(\alpha, g)$ for some $\alpha\in G^n, g\in G$, then the decoder should always correctly return $\sum_id_i\alpha_i$.

The code will depend on two external parameters: an error parameter $\delta\in (0, 1]$ and a range parameter $r\in \mathbb{N}$. If a string $\beta\in G^m$ is such that the test passes with suitably high probability (precisely $1-\frac14\delta$), then there should be some $\alpha\in G^n$ such that the decoder evaluates $\sum_id_i\alpha_i$ with probability at least $1-\delta$ for any valid decoding parameter $d$. A value of $d$ is valid if $\Vert d\Vert_1\le r$ (here $\Vert-\Vert_1$ refers to the $L_1$-length, so $\Vert d\Vert_1=\sum_i|d_i|$).

The important point here is that $\alpha$ only depends on $\beta$ (and not on the decoding parameter $d$). If the test passes with high probability, then $\beta$ is ``close'' to an encoding of $\alpha$ in some sense (but not necessarily in the Hamming distance sense). 

\paragraph{Defining the tester and decoder.}

In addition to $n$, our code has two internal parameters $k_1$ and $k_2$ with $k_1>k_2$. We will later determine $k_1$ and $k_2$ based on $r$ and $\delta$. We set $m=(k_1+1)^n$ and identify $\{1, \dots, m\}$ with $\{0, \dots, k_1\}^n$. 

Consider an element of $G^m$, viewed as a map $A:\{0, \dots, k_1\}^n\rightarrow G$. To test the code, we choose two elements $x_1$ and $y_1$ independently and uniformly at random from $\{0, \dots, k_1\}^n$. We then find some $x_2$ and $y_2$ in $\{0, \dots, k_1\}^n$ such that $x_1+x_2=y_1+y_2$, and we check that $A(x_1)+A(x_2)=A(y_1)+A(y_2)$. 

We choose $x_2$ and $y_2$ as follows. First, choose a value $z\in \{-k_2, \dots, k_2\}^n$ uniformly at random and set $x_2=k_1\mathbf{1}+z-x_1$ and $y_2=k_1\mathbf{1}+z-y_1$ (here $\mathbf{1}$ is the all-ones vector). If $x_2$ or $y_2$ is not in $\{0, \dots, k_1\}^n$, then the test passes automatically.

We call this the \emph{midpoint test}, since the line segments from $x_1$ to $x_2$ and from $y_1$ to $y_2$ intersect at their shared midpoint $\frac12(x_1+x_2)=\frac12 (y_1+y_2)$. The test uses the values of $A(x_1)$, $A(x_2)$, $A(y_1)$ and $A(y_2)$ to linearly interpolate $A$ at the midpoint in two ways, and checks that the two values agree.

To decode the code, we choose $x$ uniformly at random from the set of points in $\{0, \dots, k_1\}^n$ such that $x+d\in \{0, \dots, k_1\}^n$. We then evaluate $A(x+d)-A(x)$.

If $\alpha\in G^n$ and $A(x)=\sum_i\alpha_ix_i+c$, then it is straightforward to see that the test passes with probability $1$, and that the decoder always correctly decodes $\sum_id_i\alpha_i$. 

\paragraph{Verifying the tester and decoder.}

We first specify the values of $k_1$ and $k_2$ in terms of $\delta$ and $r$. We set:

\[k_2=4\lceil\delta^{-1}\rceil r,\quad k_1=8\lceil\delta^{-1}\rceil nk_2=32\lceil\delta^{-1}\rceil^2nr\]

The following lemma will be used in the verification of the tester and decoder.

\begin{lemma}\label{lem:gridfraction}
Let $k\in \mathbb{N}$. Suppose that $x$ is chosen uniformly at random in $\{0, \dots, k\}^n$ and let $d\in \mathbb{Z}^n$. Then the probability that $x+d\in \{0, \dots, k\}^n$ is at least $1-\frac{\Vert d\Vert_1}{k+1}$
\end{lemma}

\begin{proof}
The claim is vacuously true when $\Vert d\Vert_1>k$, so suppose that $\Vert d\Vert_1\le k$

Write $x=(x_1, \dots, x_n)$ and $d=(d_1, \dots, d_n)$. The probability that $x_i\in \{0, \dots, k\}$ is $1-\frac{|d_i|}{k+1}$. So the probability that $x\in \{0, \dots, k_1\}^n$ is at least:

\[1-\sum_{i=1}^n\frac{|d_i|}{k+1}=1-\frac{\Vert d\Vert_1}{k+1}\]
\end{proof}

In the rest of this section, we fix an input $A: \{0, \dots, k_1\}^n\rightarrow G$. For each $d\in \mathbb{Z}^n$, let $B(d)$ be the most common value of $A(x+d)-A(x)$ for $x$ chosen uniformly at random in $\{0, \dots, k_1\}^n$, ignoring the values where $x+d\notin \{0, \dots, k_1\}^n$.

\begin{lemma}\label{lem:Bdefinedness}
If the midpoint test passes with probability at least $1-\frac{1}{4}\delta$ (on input $A$) and $d$ satisfies $\Vert d\Vert_1\le r$, then both $x+d\in \{0, \dots, k_1\}^n$ and $A(x+d)-A(x)=B(d)$ hold for at least a $1-\delta$ fraction of $x\in \{0, \dots, k_1\}^n$
\end{lemma}

\begin{proof}
Let $\epsilon=\frac14 \delta$.

For $x_1$ and $y_1$ chosen independently and uniformly at random in $\{0, \dots, k_1\}^n$, we want a lower bound on the probability that there are some $x_2, y_2$ such that both of:

\begin{align}
A(x_1)+A(x_2)&=A(y_1)+A(y_2)\label{eq:basetest}\\
A(x_1+d)+A(x_2)&=A(y_1+d)+A(y_2)\label{eq:shifttest}
\end{align}

\noindent hold.

We say that the test passes for a triple $(x, y, z)$ if the test passes when choosing those random values. That is, if $k_1\mathbf{1}+z-x$ or $k_2\mathbf{1}+z-y$ is out of range or $A(x)+A(k_1\mathbf{1}+z-x)=A(y)+A(k_1\mathbf{1}+z-y)$. 

Choose $z$ uniformly at random in $\{-k_2, \dots, k_2\}^n$ and set $x_2=k_1\mathbf{1}+z-x_1$ and $y_2=k_1\mathbf{1}+z-y_1$. So \eqref{eq:basetest} holds when $x_2$ and $y_2$ are in range and the test passes for $(x_1, y_1, z)$. 

We can write $x_2=k_1\mathbf{1}+(z+d)-(x_1+d)$ and $y_2=k_1\mathbf{1}+(z+d)-(y_1+d)$. So \eqref{eq:shifttest} holds when $x_1+d, y_1+d\in \{0, \dots, k_1\}^n$, $z+d\in \{-k_2, \dots, k_2\}^n$, $x_2, y_2\in \{0, \dots, k_1\}^n$, and the test passes for $(x_1+d, y_1+d, z+d)$.

If we choose $z\in \{-k_2, \dots, k_2\}$, then \eqref{eq:basetest} and \eqref{eq:shifttest} hold if all of the following are true:

\begin{itemize}
    \item $(x_1, y_1, z)$ passes,
    \item $x_2$ and $y_2$ are in range,
    \item $(x_1+d, y_1+d, z+d)$ passes or at least one of $x_1+d$, $y_1+d$, or $z+d$ is out of range, and
    \item $x_1+d$, $y_1+d$, and $z+d$ are in range
\end{itemize}

We now want a bound on the probabilities of each of these conditions. 

The probability that $(x_1, y_1, z)$ passes is at least $1-\epsilon$ by assumption. Note that $\Vert z\Vert_1\le nk_2$, so by \Cref{lem:gridfraction}, $x_2$ and $y_2$ are each in range with probability at least $1-\frac{nk_2}{k_1}$. So both are in range with probability at least $1-2\frac{nk_2}{k_1}\ge 1-\frac14\delta=1-\epsilon$. 

We now claim that, with probability at least $1-\epsilon$, either $(x_1+d, y_1+d, z+d)$ passes or one of $x_1+d$, $y_1+d$, and $z+d$ is out of range. Indeed, the triple:

\[(x_1+d \text{ mod } (k_1+1), y_1+d\text{ mod } (k_1+1), ((z+d+k_2\mathbf{1})\text{ mod } (2k_2+1))-k_2\mathbf{1})\]

\noindent has the same distribution as $(x_1, y_1, z)$, so passes with probability at least $1-\epsilon$. The values in the new triple are equal to $(x_1+d, y_1+d, z+d)$ when $x_1+d$, $y_1+d$, and $z+d$ are in range.

By \Cref{lem:gridfraction}, $z+d\in \{-k_2, \dots, k_2\}^n$ with probability at least $1-\frac{\Vert d\Vert_1}{2k_2}$. Each of $x_1+d$ and $y_1+d$ has a probability at least $1-\frac{\Vert d\Vert_1}{k_1}$ of being in $\{0, \dots, k_1\}^{n}$. So with probability at least $1-2\frac{\Vert d\Vert_1}{k_1}-\frac{\Vert d\Vert_1}{2k_2}\ge 1-\epsilon$, all of $x_1+d$, $y_1+d$, and $z+d$ are in range. 

Altogether, we see that \eqref{eq:basetest} and \eqref{eq:shifttest} hold with probability at least $1-4\epsilon=1-\delta$. By subtracting \eqref{eq:basetest} from \eqref{eq:shifttest}, we find that $A(x_1+d)-A(x_1)=A(y_1+d)-A(y_1)$ with probability at least $1-\delta$. 

In particular, the most common value $B(d)$ of $A(x+d)-A(x)$ occurs at least a $1-\delta$ fraction of the time.
\end{proof}

So if $x$ is chosen uniformly at random, then the probability that $A(x+d)-A(x)=B(d)$ given that $x+d\in \{0, \dots, k_1\}^n$ is at least $1-\delta$. We should now check that there is some $\alpha\in \mathbb{R}^n$ such that $B(d)=\sum_{i}\alpha_id_i$. 

\begin{lemma}\label{lem:Blinearity}
Suppose that $\delta\le \frac14$ and the midpoint test passes (on input $A$) with probability at least $1-\frac14\delta$. Then for $d_1, d_2$ and $d_3$ such that $d_1+d_2=d_3$ and $\Vert d_i \Vert_1\le r$ for $i=1,2,3$, we have $B(d_1)+B(d_2)=B(d_2)$.
\end{lemma}

\begin{proof}
Choose $x$ uniformly at random in $\{0, \dots, k_1\}$. By \Cref{lem:Bdefinedness}, $A(x+d_1)-A(x)=B(d_1)$ with probability at least $\frac{3}{4}$ and $A(x+d_3)-A(x)=B(d_3)$ with probability at least $\frac{3}{4}$. 

With probability at least $\frac{3}{4}$, we have $A((x+d_1\mod k_1+1)+d_2)-A(x+d_1\mod k_1+1)=B(d_2)$. By \Cref{lem:gridfraction}, we have $x+d_1\in \{0, \dots, k_1\}^n$ with probability at least $1-\frac{\Vert d_1\Vert}{k_1}\ge\frac{7}{8}$. So with probability at least $\frac{5}{8}$, we have $A(x+d_1+d_2)-A(x+d_1)=B(d_2)$.

So with probability at least $\frac{1}{8}>0$, we have $A(x+d_1)-A(x)=B(d_1)$, $A(x+d_3)-A(x)=B(d_3)$, and $A(x+d_1+d_2)-A(x+d_1)=B(d_2)$. Since $d_1+d_2=d_3$, $B(d_1)+B(d_2)=B(d_3)$.
\end{proof}

If $r\ge 1$ and $e_i$ is the $i$th standard basis vector, then repeated applications of \Cref{lem:Blinearity} show that $B(d)=\sum_id_iB(e_i)$ for $\Vert d\Vert_1\le r$. 

\paragraph{Limitations of the tester.}

In comparison to the standard notion of a locally testable code, the test passing with probability $1-\frac14\delta$ doesn't actually guarantee that $A$ is close to being a codeword, at least in the Hamming distance sense.

For example, let $S$ be the set of points in $\{0, \dots, k_1\}^n$ where the first coordinate is at most $\frac12k_1$. Let $G$ be the additive group of the reals and consider the map $A:\{0, \dots, k_1\}^n\rightarrow \mathbb{R}$ given by $A(x)=0$ if $x\in S$ and $A(x)=1$ otherwise.

For large enough $k_1/k_2$, this string passes the test with high probability. Indeed, when we choose $x_1, x_2, y_1$ and $y_2$, it is very likely that exactly one of $x_1$ or $x_2$ is in $S$ and exactly one of $y_1$ or $y_2$ is in $S$, so $A(x_1)+A(x_2)=1=A(y_1)+A(y_2)$. However, $A$ differs from the nearest valid codeword in half it its locations.

By \Cref{lem:Bdefinedness,lem:Blinearity}, we know that the decoder is still likely to correctly decode $B(d)=\sum_id_iB(e_i)$. In this example, it is not too hard to see the decoder will almost always output $A(x+d)-A(x)=0$, since $x+d\in S$ for most $x\in S$. 

\section{Using the midpoint code to check constraints}\label{sec:constraints}

\subsection{Checking linear constraints}\label{sec:lineartest}

Suppose we have a $\mathbb{Z}$-linear map $f:G^d\rightarrow G^n$, that is a map that sends $(g_1, \dots, g_d)$ to the element whose $i$th entry is is $\sum_{j=1}^dM_{ij}g_j$ for some $\mathbb{Z}$-valued $n\times d$ matrix $M$. 

Given an element $\alpha$ of $G^n$, we can use the midpoint test to check if $\alpha$ is close to an element in the image of $f$. We use a midpoint code over $G$ with dimension $d$, $\delta=\frac14$, and $r=\text{max}_{i\in \{1, \dots, n\}}(\Vert M^Te_i\Vert_1)$, where $e_i$ is the $i$th standard basis vector. 

Given an element $\alpha\in G^n$ and $A:\{0, \dots, k\}^d\rightarrow G$, we can test that $\alpha$ is close to being in the image of $f$ as follows: with probability $\frac12$, run the midpoint test on $A$. Otherwise, choose a random $i\in \{1, \dots, n\}$ and $x\in \{0, \dots, k_1\}^d$ such that $x+Me_i\in \{0, \dots, k_1\}^d$ and check that:

\[\alpha_i+A(x)=A(x+M^Te_i)\]

We call the combined test the \emph{linear test}. 

\begin{lemma}\label{lem:lineartest}
If the linear test passes with probability $1-\epsilon$ with $\epsilon\le \frac18\delta=\frac1{32}$, then there is some $\beta\in \text{img}(f)$ such that $\alpha_i=\beta_i$ for at least a $1-4\epsilon$ fraction of $i\in \{1, \dots, n\}$.
\end{lemma}

\begin{proof}
The midpoint test occurs with probability $\frac12$, so passes with probability at least $1-2\epsilon\ge 1-\frac{1}4\delta$. So for $\Vert d\Vert_1\le r$, the most common value $B(d)$ of $A(x+d)-A(x)$ occurs at least a $1-\delta=\frac34$ fraction of the time by \Cref{lem:Bdefinedness}.

Let $\beta_i=B(Me_i)$ for each $i=1, \dots, n$. By \Cref{lem:Blinearity}, $B(M^Te_i)=\sum_{j=1}^dM_{ij}B(e_j)=f(B(e_1), \dots, B(e_d))$. So $\beta\in \text{img}(f)$.

Let $S\subseteq \{1, \dots, n\}$ be the set of indices $i$ for which $\alpha_i=A(x+Me_i)-A(x)$ for at least half of the values of $x\in \{0, \dots, k_1\}^d$ satisfying $x+Me_i\in \{0, \dots, k_1\}^d$. The second test passes with probability at least $1-2\epsilon$, so $S$ contains at least a $1-4\epsilon$ fraction of $i\in \{1, \dots, n\}$.

Now $\Vert Me_i\Vert_1\le r$ and $k_1\ge 8r$, so $x+Me_i\in \{0, \dots, k_1\}^d$ for at least a $\frac{7}{8}$ fraction of $x\in \{0, \dots, k_1\}^d$ by \Cref{lem:gridfraction}. So $\alpha_i=A(x+Me_i)-A(x)$ for at least $\frac{7}{16}$ of $x\in \{0, \dots, k_1\}^d$. So for $i\in S$, we have $\alpha_i=A(x+Me_i)-A(x)$ and $A(x+Me_i)-A(x)=B(Me_i)=\beta_i$ for at least a $\frac{3}{16}$ fraction of $x\in \{0, \dots, k_1\}^d$.

So $\alpha_i=\beta_i$ for $i\in S$. Since $\beta\in \text{img}(f)$ and $S$ contains at least a $1-4\epsilon$ fraction of $i\in \{1, \dots, n\}$, this proves the claim. 
\end{proof}

We will want to test membership in a subset of $\mathbb{R}^n$ cut out by some constraints, rather than one parameterized by a function $f$. If this subset is a linear subspace, then it is straightforward to convert a description in terms of integer-coefficient linear constraints into a parameterization. 

\begin{lemma}\label{lem:basisfinding}
Let $V$ be $d$-dimensional linear subspace of $\mathbb{R}^n$ given by the kernel of a matrix $B$ with integer entries, where each row $c$ of $B$ has $\Vert c\Vert_1\le k$ for some integer $k\ge 1$. 

Then there is a basis $b_1, \dots, b_d$ for $V$ where the entries of the $b_i$ are integers with absolute value at most $k^{n-1}$. Furthermore, there is some positive integer $D$ and some function $\rho:[d]\rightarrow [n]$ such that, for $\alpha\in V$, $\alpha=\frac{1}{D}\sum_{i=1}^d\alpha_{\rho(i)}b_i$.
\end{lemma}

\begin{proof}
By removing redundant constraints, we can suppose that $B$ is a $(n-d)\times n$ non-singular matrix, where $d$ is the dimension of $V$.  

We need to compute a basis for $V$. First, we add $d$ rows to $B$, where the $i$th new row has a $1$ in position $\rho(i)$ and is $0$ in all other positions, to obtain a non-singular $n\times n$ matrix $M$. By multiplying one of the constraint rows by $-1$ if necessary, we can suppose that $\text{det}(M)$ is positive.

Now $\text{det}(M)>0$ is an integer with absolute value at most $k^n$, and $\text{det}(M)M^{-1}=\text{adj}(M)^T$ is matrix with integer entries, each of absolute value at most $k^{n-1}$. The $n-d+1$ through $n$th columns of $\text{det}(M)M^{-1}$ give a basis $b_1, \dots, b_d$ for $V$.

Given an element $\alpha\in V$, we have $\alpha=M^{-1}M\alpha=\frac{1}{\text{det}(M)}\sum_{i=1}^d\alpha_{\rho(i)}b_i$.
\end{proof}

The multiplicative group of the positive reals is isomorphic to the additive group of the reals, with the isomorphism being given by the logarithm. \Cref{lem:threeconstraintconversion} produces constraints of form $x+y=z$, $x=q$, and $(1+x)(1+y)=1+z$, where $q$ is a constant. The first two types of constraints cut out a linear subspace of $\mathbb{R}^{n+m}$, while the third type becomes linear after taking the logarithm, since $\log((1+x)(1+y))=\log(1+x)+\log(1+y)$. 

\subsection{Checking general constraints}\label{sec:constrainttest}

In order to check more general types of constraints, we use the following lemma, which is proved in \Cref{sec:formulaconversion}:

\begin{lemma}\label{lem:threeconstraintconversion}
For each $n\in \mathbb{N}$ there is some $m\in \mathbb{N}$, such that for any $q\in (0, \frac12]$ and a constraint $V\in \mathcal{P}\left(\mathcal{C}(q)_{[\frac12, 2]}^{n}\right)$, we can construct a subset $V'$ of $\mathbb{R}^{n+m}$ such that:

\begin{itemize}
    \item $V'$ is cut out by constraints of form $x+y=z$, $(1+x)(1+y)=(1+z)$, and $x=q$, with exactly one constraint of form $x=q$,
    \item for any $x\in V$, there is a $y\in [-6q, 6q]^m$ such that $(qx, y)\in V'$, and
    \item if $(x, y)\in V'$ for some $x\in \mathbb{R}^n$, then $q^{-1}x\in V$
\end{itemize}
\end{lemma}

Let $V\subseteq [-2, 2]^n$ be an element of $\mathcal{P}\left({\mathcal{C}(q)_{[-2, 2]}^n}\right)$. We want to create a single test that checks for membership in $V$. Using \Cref{lem:threeconstraintconversion}, we create a new set $V'\subseteq \mathbb{R}^{n+m}$, where $V'\subseteq [-6q, 6q]^{n+m}$ and the projection of $V'$ on to the first $n$ variables is $qV$. Write $V'=V_1\cap V_2\cap V_3$, where $V_1$ is cut out by constraints of form $\alpha_i+\alpha_j=\alpha_k$, $V_2$ is cut out by constraints of form $(\alpha_i+1)(\alpha_j+1)=\alpha_k+1$, and $V_3$ is cut out by a single constraint of form $\alpha_i=q$. We choose $q$ small enough so that $6q\le \frac12$. 

We encode an element of $V'$ as a map $\{0, 1\}^{n+m}\rightarrow \mathbb{R}$. For each $\sigma:\{1, \dots, n+m\}\rightarrow \{0, 1\}$, we create a variable $u_\sigma$. To encode $\alpha\in V'$, we would set $u_\sigma=\prod_{i=1}^{n+m}(\alpha_i+1)^{\sigma_i}$. 

Note that $6q\le \frac12$ and $V'\subseteq [-6q, 6q]$, so if $\alpha\in V'$ then each $\alpha_i+1$ is positive. So we will assume (for now) that the $u_\sigma$ are positive.

Let $V_2'$ be the set of $(u_\sigma)_{\sigma\in \{0, 1\}^{n+m}}$ such that $u_\sigma=\prod_{i=1}^{n+m}(\alpha_i+1)^{\sigma_i}$ for some $\alpha\in V_2$. We first want to test that $(u_\sigma)_{\sigma\in \{0, 1\}^{n+m}}$ is close to an element of $V_2'$.

A constraint of form $(\alpha_i+1)(\alpha_j+1)=\alpha_k+1$ becomes linear after taking the logarithm. So by \Cref{lem:basisfinding}, there is an integer-valued matrix $M$ such that every $\alpha\in V_2$ has $\alpha_i+1=\prod_{j=1}^dg_j^{M_{ij}}$ for some positive real numbers $g_1, \dots, g_d$.

The expression $u_\sigma=\prod_{i=1}^{n+m}(\alpha_i+1)^{\sigma_i}$ then gives us a parameterization of $V_2'$ in terms of the $g_j$. So we can use a linear test to check that $(u_\sigma)_{\sigma\in \{0, 1\}^{n+m}}$ is close to an element of $V_2'$. If the test passes with probability $1-\epsilon$ for $\epsilon\le \frac1{32}$, then there is some $\alpha\in V_2$ such that $u_\sigma=\prod_{i=1}^{n+m}(\alpha_i+1)^{\sigma_i}$ for at least a $1-4\epsilon$ fraction of $\sigma$. 

In addition to the values of $u_\sigma$, the test takes as input an auxiliary string $A_1:\{0,\dots, a_1\}^{d_1}\rightarrow \mathbb{R}$. It performs tests of form $A_1(x_1)A_1(x_2)=A_1(y_1)A_1(y_2)$ or $A(x+d)=A(x)u_\sigma$.

This tests membership in $V_2$. We now need tests to test membership in $V_1$ and $V_3$. 

$V_3$ is cut out by a single constraint $x_i=q$. To test membership in $V_3$, we choose a random $\sigma$ such that $\sigma(i)=0$ and test that $u_{\sigma+e_i}=(1+q)u_{\sigma}$.

To test membership in $V_1$, we define a linear subspace $V_1'$ of assignments of the values of $u_\sigma$. For each constraint $\alpha_i+\alpha_j=\alpha_k$ defining $V_1$ and each $\sigma\in \{0, 1\}^{n+m}$ with $\sigma(i)=\sigma(j)=\sigma(k)=0$, we create a constraint $u_{\sigma+e_i}+u_{\sigma+e_j}=u_{\sigma+e_k}+u_{\sigma}$, and we let $V_1'$ be the subspace cut out by all these constraints. We can create a basis for $V_1'$ using \Cref{lem:basisfinding}. 

We could use another linear test to check that $(u_\sigma)_{\sigma\in \{0, 1\}^{n+m}}$ is close to an element of $V_1'$. However, we need to be careful: this test would perform tests of form $A_2(x+d)=A_2(x)+u_\sigma$. When $q$ is small, the $u_\sigma$ are close to $1$, so this would require many entries of $A_2$ to differ by a large amount. We want to be able to keep the entries of $A_2$ in $[\frac12, 2]$, so this is not good. 

We notice that $V_1'$ is invariant under adding a constant to each $u_\sigma$. So we use instead use a linear test to check that $(u_\sigma-u_0)_{\sigma\in \{0, 1\}^{n+m}}\in V_1'$ (here $u_0$ is the $u$-variable corresponding to the all-zeros element of $\{0, 1\}^{n+m}$). In addition to the $u_\sigma$, this test takes as input an auxiliary string $A_2:\{0, \dots, a_2\}^{d_2}\rightarrow \mathbb{R}$. It performs tests of form $A_2(x_1)+A_2(x_2)=A_2(y_1)+A_2(y_2)$ or $A_2(x+d)+u_0=A_2(x)+u_\sigma$.

If $(u_\sigma-u_0)_{\sigma\in \{0, 1\}^{n+m}}\in V_1'$ is close to an element $v$ of $V_1'$, then $(u_\sigma)_{\sigma\in \{0, 1\}^{n+m}}\in V_1'$ is just as close to $(v_\sigma+u_0)_{\sigma\in \{0, 1\}^{n+m}}$, which is also in $V_1'$. 

This gives us three tests, one for each of $V_1$, $V_2$ and $V_3$. The \emph{constraint test} picks one of these tests at random and runs it.

\begin{lemma}\label{lem:constrainttest}
Suppose that $u_\sigma$, $A_1$ and $A_2$ are such that each entry is in $[\frac12, 2]$ and the constraint test passes with probability $1-\epsilon$ with $\epsilon<\frac{1}{192}$. Then there is some $\alpha\in V'$ such that $u_\sigma=\prod_{i=1}^{n+m}(1+\alpha_i)^{\sigma_i}$ for at least a $1-12\epsilon$ fraction of $\sigma\in \{0, 1\}^{n+m}$. 
\end{lemma}

\begin{proof}
Each of the three component tests pass with probability at least $1-3\epsilon>1-\frac{1}{64}$, so (using \Cref{lem:lineartest}):

\begin{itemize}
    \item there is some $\alpha=(\alpha_1, \dots, \alpha_{n+m})\in V_2$ such that $u_\sigma=\prod_{i=1}^{n+m}(1+\alpha_i)^{\sigma_i}$ for at least a $1-12\epsilon$ fraction of $\sigma:\{1,\dots, n+m\}\rightarrow \{0, 1\}$,
    \item there is some $(v_\sigma)_{\sigma \in \{0, 1\}^{n+m}}\in V_1'$ such that $u_\sigma=v_\sigma$ for at least a $1-12\epsilon$ fraction of $\sigma$, and
    \item for the value of $i$ corresponding to the constraint $\alpha_i=q$, $u_{\sigma+e_i}=(1+q)u_{\sigma}$ for at least a $1-3\epsilon$ fraction of $\sigma$ such that $\sigma(i)=0$. 
\end{itemize}

We want to show $\alpha\in V'$. Let $S$ be the set of values of $\sigma$ such that $u_\sigma=\prod_{i=1}^{n+m}(1+\alpha_i)^{\sigma_i}$, and $(v_\sigma)_{\sigma\in \{0, 1\}^{n+m}}\in V_1'$. By the above, $S$ contains a least a $1-24\epsilon$ fraction of all $\sigma\in \{0, 1\}^{n+m}$. 

For each constraint of form $\alpha_i+\alpha_j=\alpha_k$ cutting out $V_1$, we claim there is some $\sigma$ such that $\sigma(i)=\sigma(j)=\sigma(k)=0$ and ${\sigma+e_i}, {\sigma+e_j}, {\sigma+e_k}$ and ${\sigma}$ are all in $S$. 

Suppose not. Then for each $\sigma$ such that $\sigma(i)=\sigma(j)=\sigma(k)=0$, one of $\sigma$, $\sigma+e_i$, $\sigma+e_j$, or $\sigma+e_k$ is not in $S$. So the set $\{\mu\in S:\mu\text{ and }\sigma \text{ agree on indices other than }i,j,k\}$ has size at most $7$. Adding up over all $\sigma$ with $\sigma(i)=\sigma(j)=\sigma(k)=0$, we see that $2^{-n-m}|S|\le \frac{7}{8}$, which is a contradiction since $24\epsilon<\frac18$. 

Since ${\sigma+e_i}, {\sigma+e_j}, {\sigma+e_k}$ and ${\sigma}$ are in $S$, we have that $u_{\sigma+e_i}+u_{\sigma+e_j}=u_{\sigma+e_k}+u_{\sigma}$. Also, $u_{\sigma+e_i}=(1+\alpha_i)u_\sigma$, $u_{\sigma+e_j}=(1+\alpha_j)u_\sigma$, and $u_{\sigma+e_k}=(1+\alpha_k)u_\sigma$. So $(\alpha_i+1)u_\sigma+(\alpha_j+1)u_\sigma=(\alpha_k+1)u_\sigma+u_{\sigma}$. Since $u_\sigma>0$, this implies that the constraint $\alpha_i+\alpha_j=\alpha_k$ is satisfied. This is true for every constraint $\alpha_i+\alpha_j=\alpha_k$ cutting out $V_1$, so $\alpha\in V_1$.

Finally, let $i$ be the index such that a constraint $\alpha_i=q$ cuts out $V_3$. Similarly to above, we let $S$ be the set of $\sigma\in \{0, 1\}^{n+m}$ such that $u_\sigma=\prod_{j=1}^{n+m}(1+\alpha_j)^{\sigma_j}$ and $u_{\sigma+e_i}=(1+q)u_\sigma$ if $\sigma(i)=0$, or $u_{\sigma}=(1+q)u_{\sigma-e_i}$ if $\sigma(i)=1$. Now $S$ contains at least a $1-15\epsilon$ fraction of $\sigma\in \{0, 1\}^{n+m}$. 

So if $15\epsilon<\frac12$, then there is some $\sigma$ such that $\sigma(i)=0$ and both $\sigma$, $\sigma+e_i$ are in $S$. So $1+q=1+\alpha_i$, so $\alpha\in V_3$. 

So if $\epsilon<\frac{1}{8\cdot24}=\frac1{192}$, then $\alpha\in V_1\cap V_2\cap V_3=V'$.
\end{proof}

We now need to check that each $\alpha\in V'$ has a corresponding assignment of auxiliary variables that satisfies every possible test. We need to ensure that this assignment has all variables in $[\frac12, 2]$, which requires choosing $q$ sufficiently small. 

\begin{lemma}\label{lem:constrainttestsatisfying}
There is some value of $q>0$ depending only on $n$ such that, for any $\alpha\in V'$, we can set $u_\sigma=\prod_{i=1}^{n+m}(1+\alpha_i)^{\sigma_i}$ for each $\sigma\in \{0, 1\}^{n+m}$ and assign values of $A_1$, $A_2$ such that:

\begin{itemize}
    \item each $u_\sigma$ and each entry of $A_1$ or $A_2$ is in $[\frac12, 2]$, and
    \item the constraint test passes with probability $1$
\end{itemize}
\end{lemma}

\begin{proof}
From the proof of \Cref{lem:threeconstraintconversion}, we have $m=1+22n$. 

In order to have $u_\sigma=\prod_{i=1}^{n+m}(1+\alpha_i)^{\sigma_i}\in [\frac12, 2]$, it is sufficient to have $|\alpha_i|\le \frac{\log(2)}{\log(n+m)}$ for each $i$. Since $V'\subseteq [-6q, 6q]^{n+m}$, it is sufficient to have $6q\le \frac{\log(2)}{\log(n+m)}$.

To define the linear test for $V_2'$, we first found a basis for $\log(1+V_2)=\{(\log(1+\alpha_1), \dots, \log(1+\alpha_{n+m}))|\alpha\in V_2\}$, and then parameterized $V_2'$ in terms of $V_2$. By \Cref{lem:basisfinding}, the vectors in the basis for $\log(1+V_2)$ have entries with absolute value at most $3^{m+n-1}$. Each entry of $\log(1+V_2')$ is a sum of at most $n+m$ entries of of $\log(1+V_2)$, so the corresponding basis for $\log(1+V_2')$ has entries of absolute value at most $(n+m)3^{n+m-1}$. The code used then needs range parameter at most $(n+m)^23^{n+m-1}$. 

Since this code has dimension $d_1$ and $\delta=\frac14$, the width of the grid is $a_1=512d_1(n+m)^23^{n+m-1}\le 512(n+m)^33^{n+m-1}$. To encode $\alpha\in V'$, we should have $A_1(x)=\prod_{i=1}^{d_1}(1+\alpha_{\rho_1(i)})^{x_i/D_i}$, where $D_1$ and $\rho_1$ come from \Cref{lem:basisfinding}. In order for this to be in $[\frac12, 2]$, it is sufficient to have $|\alpha_i|\le \frac{\log(2)}{512(n+m)^43^{n+m-1}}$ for each $i$, so it is sufficient to have $6q\le \frac{\log(2)}{512(n+m)^43^{n+m-1}}$. 

Now $V_1'$ is a linear subspace of $\mathbb{R}^{2^{n+m}}$ cut out by constraints of form $u_{\sigma+e_i}+u_{\sigma+e_j}=u_{\sigma+e_k}+u_{\sigma}$. The basis found by \Cref{lem:basisfinding} has entries with absolute value at most $4^{2^{n+m}-1}$. We then need a code with range parameter at most $2^{n+m}4^{2^{n+m}-1}$. So $a_2\le 512\cdot 2^{2(n+m)}4^{2^{n+m}-1}$.

To encode $\alpha\in V'$, we set:

\[A_2(x)=1+\sum_{i=1}^{d_2}x_i\frac{1}{D_2}\left(\prod_{j=1}^{n+m}(1+\alpha_j)^{\rho_2(i,j)}-1\right)\]

Here $D_2$ and $\rho_2$ come from \Cref{lem:basisfinding}, so $\rho_2$ is a map $[d_2]\rightarrow \{0, 1\}^{n+m}$, viewed here as a map $[d_2]\times [n+m]\rightarrow \{0, 1\}$. 

We want each $A_2(x)$ to be in $[\frac12, 2]$. It is sufficient to have:

\[d_2\cdot 512\cdot 2^{2(n+m)}4^{2^{n+m}-1}\cdot \left((1+6q)^{n+m}-1\right)\le \frac12\]

This is satisfied when $q\le2^{-2^{n+m+\mathcal{O}(1)}}$. So when $q$ is sufficiently small, all the auxiliary variables can be assigned in $[\frac12, 2]$. The bound on $q$ depends only on $n$ (since $m=1+22n$). 
\end{proof}

\section{The assignment tester}\label{sec:finalproof}

We are now ready to prove \Cref{lem:main}. Recall that we want to construct that, for each $n$ and sufficiently small $q$, an assignment tester for $\mathcal{P}\left(\mathcal{C}(q)_{[\frac12, 2]}^{2n}\right)$ over $[\frac12, 2]^n$ with length $\ell=4$, error threshold $\zeta=\frac{1}{1000}$, output alphabet $[\frac12, 2]$, and output constraint set $\mathcal{C}(q)_{[\frac12, 2]}^4$. We need to specify the remaining parameters of the assignment tester.

Let $V\in \mathcal{P}\left(\mathcal{C}(q)_{[\frac12, 2]}^{2n}\right)$ be the constraint that we want to test. 

\paragraph{Error correcting function and blow-up.}

The blow-up size is $2^n$, and the error correcting function is the function $f$ that maps $\alpha\in [\frac12, 2]^n$ to the string $(a_\sigma)_{\sigma\in \{0, 1\}^n}$ where $a_\sigma=\prod_{i=1}^n(1+q\alpha_i)^{\sigma_i}$. We choose $q$ smaller than $\frac{\log(2)}{2n}$, so this assigns $a_\sigma$ a value in $[\frac12, 2]$.

The following corollary of the DeMillo-Lipton-Schwartz-Zippel \cite{DL78,Zippel79,Schwartz80} Lemma is used to compute the distance of the error-correcting function:

\begin{lemma}\label{lem:schwartzzippel}
Let $R\subseteq \mathbb{R}$ be finite and let $f$ and $g$ be two different affine functions $R^n\rightarrow \mathbb{R}$. Then $f$ and $g$ agree on at most $|R|^{n-1}$ points in $R^n$. 
\end{lemma}

Using the isomorphism between the additive and multiplicative groups of the reals, this shows the error correcting function has distance at least $\frac{1}{2}$. 

\paragraph{Auxiliary variables.}

Let $m=1+22n$ be the number of auxiliary variables needed by \Cref{lem:threeconstraintconversion} applied to $\mathcal{P}\left(\mathcal{C}(q)_{[\frac12, 2]}^{2n}\right)$. The assignment tester then uses the following auxiliary variables:

\begin{itemize}
    \item the variables $(u_\sigma)_{\sigma\in \{0, 1\}^{2n+m}}$, $A_1$ and $A_2$ from the constraint test in \Cref{sec:constrainttest}, and
    \item two auxiliary strings $L_1$ and $L_2$, each viewed as maps $\{0, \dots, 512n^2\}^n\rightarrow [\frac12, 2]$
\end{itemize}

The lengths of $A_1$ and $A_2$ depend on $V$, but are bounded by a function of $n$. 

\paragraph{The tests.}

We write $(a_\sigma)_{\sigma\in \{0, 1\}^n}$ for the first set of input variables and $(b_\sigma)_{\sigma\in \{0, 1\}^n}$ for the second set of input variables. 

Now let $V$ be a constraint in $\mathcal{P}\left(\mathcal{C}(q)_{[\frac12, 2]}^{2n}\right)$, and let $V'$ be the subset of $\mathbb{R}^{2n+m}$ produced by \Cref{lem:threeconstraintconversion}. In order to test membership in $V$, we need $5$ different types of test. These are:

\begin{itemize}
    \item a linear test checking that $(a_\sigma)_{\sigma\in \{0, 1\}^n}$ is close to an encoding of some $\alpha\in \mathbb{R}^n$
    \item a linear test checking that $(b_\sigma)_{\sigma\in \{0, 1\}^n}$ is close to an encoding of some $\beta\in \mathbb{R}^n$
    \item a constraint test checking that $(u_\sigma)_{\sigma\in \{0, 1\}^{2n+m}}$ is close to an encoding of some $(\alpha', \beta', \gamma)\in V'$ (here $\alpha', \beta'\in \mathbb{R}^n, \gamma\in \mathbb{R}^m$)
    \item a consistency test checking that $\alpha=\alpha'$ 
    \item a consistency test checking that $\beta=\beta'$
\end{itemize}

The first two tests are linear tests as in \Cref{sec:lineartest}, using the auxiliary strings $L_1$ and $L_2$. These tests have range parameter $r=n$, since we can parameterize the $\sum_i(1+\alpha_i)^{\sigma_i}$ in terms of the $(1+\alpha_i)$, and $\Vert \sigma\Vert_1\le n$.

The constraint test is described in detail in \Cref{sec:constrainttest}. It remains to explain how to check consistency between the $u_\sigma$ and the $a_\sigma$ and $b_\sigma$. 

The \emph{consistency test} for $\alpha$ chooses $\sigma_1\in \{0, 1\}^{2n+m}$ uniformly at random. Then, it chooses another $\sigma_2\in \{0, 1\}^{2n+m}$ where $\sigma_2(i)=\sigma_1(i)$ for $i\in \{n+1, \dots, 2n+m\}$. For each $i\in \{1, \dots, n\}$, it sets $\sigma_2(i)=\sigma_1(i)$ with probability $\frac13$ and $\sigma_2(i)=1-\sigma_1(i)$ with probability $\frac23$, with each entry being assignment independently. Let $\mu_1, \mu_2$ be restrictions of $\sigma_1, \sigma_2$ to $\{1, \dots, n\}$ respectively. The test finally checks that $a_{\mu_1}u_{\sigma_2}=a_{\mu_2}u_{\sigma_1}$.

If $(a_\sigma)_{\sigma\in \{0, 1\}^n}$ is the encoding of some $\alpha\in \mathbb{R}^n$ and $(u_\sigma)_{\sigma\in \{0, 1\}^{2n+m}}$ is the encoding of $(\alpha, \beta, \gamma)$ for some $\beta\in\mathbb{R}^n, \gamma\in \mathbb{R}^m$, then the consistency test always passes, since:

\[\frac{u_{\sigma_1}}{u_{\sigma_2}}=\prod_{i=1}^{n}\alpha_i^{\mu_1-\mu_2}=\frac{a_{\mu_1}}{a_{\mu_2}}\]

Because of the (anti-)correlation between the first $n$ entries of $\sigma_1$ and $\sigma_2$, the value of $\mu_1-\mu_2$ is uniformly distributed over $\{-1, 0, 1\}^n$. 

Similarly, there is a consistency test for $\beta$ that tests consistency between the $u_\sigma$ and $b_\sigma$. The only difference it is that is uses the $n+1, \dots, 2n$ coordinates instead of the $1, \dots, n$ coordinates. 

The assignment test chooses one of the $5$ test types at random and run a test of that type. The test queries at most $4$ values from among the $a_\sigma$, $b_\sigma$, $u_\sigma$, and entries of $A_1$, $A_2$, $L_1$, and $L_2$, and runs a test of form $x=y$, $x=(1+q)y$, $x+y=z$, $x+y=z+w$, $xy=z$, or $xy=zw$. So the output constraint set is $\mathcal{C}(q)$.

\paragraph{Verifying the assignment tester.}

To verify the assignment tester, we need to show to check completeness and soundness. Completeness is already given by \Cref{lem:constrainttestsatisfying}, so it remains to show soundness:

\begin{lemma}\label{lem:edgetest}
Let $V\in\mathcal{P}\left(\mathcal{C}(q)_{[\frac12, 2]}^{2n}\right)$. If $(a_\sigma)_{\sigma\in \{0, 1\}^n}$, $(b_\sigma)_{\sigma\in \{0, 1\}^n}$, $(u_\sigma)_{\sigma\in \{0, 1\}^{2n+m}}$, $A_1$, $A_2$, $L_1$ and $L_2$ are such that the assignment test passes with probability at least $\frac{999}{1000}$, then there is some $(\alpha, \beta)\in V$ such that $a_\sigma=\prod_{i=1}^n(1+q\alpha_i)^{\sigma_i}$ for at least $\frac1{50}$ of $\sigma\in \{0, 1\}^n$ and $b_\sigma=\prod_{i=1}^n(1+q\beta_i)^{\sigma_i}$ for at least $\frac{1}{50}$ of $\sigma\in \{0, 1\}^n$.
\end{lemma}

\begin{proof}
The constraint test passes with probability at least $\frac{199}{200}>\frac{191}{192}$, so by \Cref{lem:constrainttest}, we know that $(u_\sigma)_{\sigma\in \{0, 1\}^{2n+m}}$ is $\frac{12}{200}$-close to the encoding of some $(\alpha', \beta', \gamma)\in V'$. By \Cref{lem:threeconstraintconversion}, $(q^{-1}\alpha', q^{-1}\beta')\in V$. 

The linear test on $(a_\sigma)_{\sigma\in \{0, 1\}^n}$ passes with probability at least $\frac{199}{200}$, so $u_\sigma$ is $\frac{1}{50}$-close the encoding of some $\alpha\in \mathbb{R}^n$ by \Cref{lem:lineartest}. We want to check that $\alpha=\alpha'$.

Choose $\sigma_1$, $\sigma_2$, $\mu_1$, and $\mu_2$ as in the consistency test. The consistency test passes with probability at least $\frac{199}{200}$. When this happens, we have $a_{\mu_1}u_{\sigma_2}=a_{\mu_2}u_{\sigma_1}$.

Each of $\sigma_1$ and $\sigma_2$ is uniformly distributed over functions $\{0, 1\}^{2n+m}$ (though they aren't independent). As restrictions of $\sigma_1$ and $\sigma_2$, $\mu_1$ and $\mu_2$ are each uniformly distributed over $\{0, 1\}^n$. So with probability at least $\frac{200-2\cdot 12-2\cdot 4}{200}=\frac{168}{200}$, $u_{\sigma_1}$ and $u_{\sigma_2}$ agree with the encoding of $(\alpha', \beta', \gamma)$ and $a_{\mu_1}$ and $a_{\mu_2}$ agree with the encoding of $\alpha$. With probability at least $\frac{167}{200}$, all this holds and $a_{\mu_1}u_{\sigma_2}=a_{\mu_2}u_{\sigma_1}$. When all this happens, we have:

\[\prod_{i=1}^n(1+\alpha_i)^{\mu_1(i)-\mu_2(i)}=\prod_{i=1}^n(1+\alpha_i')^{\mu_1(i)-\mu_2(i)}\]

Now $\mu_1-\mu_2$ is uniformly distributed over $\{-1, 0, 1\}^n$. So for at least a $\frac{167}{200}>\frac13$ fraction of $\mu\in \{-1, 0, 1\}^n$, we have:

\[\prod_{i=1}^n(1+\alpha_i)^{\mu(i)}=\prod_{i=1}^n(1+\alpha_i')^{\mu(i)}\]

By \Cref{lem:schwartzzippel}, this implies that $\alpha=\alpha'$. A similar argument shows that $(b_\sigma)_{\sigma\in \{0, 1\}^n}$ is $\frac{1}{50}$-close to the encoding of $\beta'$, completing the proof.
\end{proof}

The test depends on a number of random choices. Each of these is a uniform random choice of a bounded number of alternatives, so all these choices can be determined in terms of a single uniform random integer in $[R]$ for some integer $R$. For fixed $n$, there are finitely many $V\in \mathcal{P}\left(\mathcal{C}(q)_{[\frac12, 2]}^{2n}\right)$, so by taking least common multiples, we can suppose that $R$ only depends on $n$. 

So we have constructed a family of assignment testers satisfying the requirements of \Cref{lem:main}, proving the lemma and completing the proof of \Cref{thm:main}. 

\section{Manipulating formulas}\label{sec:formulaconversion}

In this section, we prove a few lemmas that involve algebraic manipulations of formulas and constraints. We will prove that ETR-$\mathcal{C}(q)$ is $\exists\mathbb{R}$-hard, and that there is a gap-preserving reduction from MAX-ETR-$\mathcal{C}(q)$ to MAX-ETR-INV. Along the way, we will also prove a technical lemma that is useful for the construction of assignment testers. 

\paragraph{$\exists\mathbb{R}$-hardness of ETR-$\mathcal{C}(q)$.}

The $\exists\mathbb{R}$-hardness of ETR-$\mathcal{C}(q)$ follows almost immediately from a lemma in \cite{AAM2021}. 

\begin{lemma}\label{lem:cqhardness}
For each $q>0$, the problem $\mathcal{C}(q)$ is $\exists\mathbb{R}$-hard.
\end{lemma}

\begin{proof}
Abrahamsen, Adamaszek, and Miltzow \cite{AAM2021} define the problem ETR$_{[\frac12, 2]}^{1,+,\cdot}$ to be the real constraint satisfaction problem with variables in $[\frac12, 2]$ and constraints of form $x=1$, $x+y=z$, and $xy=z$. They show that this problem is $\exists\mathbb{R}$-hard (see Section 4 in \cite{AAM2021}).

An instance of ETR$_{[\frac12, 2]}^{1,+,\cdot}$ is already almost an instance of ETR-$\mathcal{C}(q)$, since the constraints $x+y=z$ and $xy=z$ can appear in an instance of ETR-$\mathcal{C}(q)$. Since $0\notin [\frac12, 2]$ we can replace a constraint $x=1$ with $x\cdot x=x$.
\end{proof}

\paragraph{Shrinking a constraint.}

We now need a technical lemma that writes a constraint in $\mathcal{C}(q)$ in terms of simpler constraints. This will be used to help us convert an instance of MAX-ETR-$\mathcal{C}(q)$ to an instance of ETR-INV, and in the proof of \Cref{lem:threeconstraintconversion}. 

\begin{lemma}\label{lem:constraintshrinking}
Let $V\subseteq \mathbb{R}^4$ be a constraint in $\mathcal{C}(q)$. Then there is a subset $V'$ of $\mathbb{R}^{74}$ such that:

\begin{itemize}
    \item $V'$ is cut out by constraints of form $x=q$, $x+y=z$, and $(1+x)(1+y)=(1+z)$, with at most $1$ constraint of form $x=q$, at most $52$ constraints of form $x+y=z$, and at most $18$ constraints of form $(1+x)(1+y)=1+z$
    \item if $x\in V$, then there is some $y\in [-6q, 6q]^{70}$ such that $(qx, y)\in V'$, and
    \item if $(x, y)\in V'$ with $x\in \mathbb{R}^4$ and $y\in \mathbb{R}^{70}$, then $q^{-1}x\in V$
\end{itemize}
\end{lemma}

\begin{proof}
We start by creating one auxiliary variable $\alpha_q$ and a constraint $\alpha_q=q$. 

If $u_1, u_2, u_3$ and $u_4$ are variables in $V'$, then we can create a constraint $u_1u_2=u_3u_4$ using $5$ auxiliary variables $\alpha_1, \alpha_2, \alpha_3, \alpha_4$ and $\alpha_5$ by creating $6$ constraints (two multiplication and four addition):

\[\begin{array}{c}
(1+u_1)(1+u_2)=(1+\alpha_1)\\
\alpha_2+u_1=\alpha_1\\
\alpha_3+u_2=\alpha_2\\
(1+u_3)(1+u_4)=(1+\alpha_4)\\
\alpha_5+u_3=\alpha_4\\
\alpha_3+u_4=\alpha_5\\
\end{array}\]

These imply:

\[\begin{array}{c}
\alpha_1=u_1u_2+u_1+u_2\\
\alpha_2=\alpha_1-u_1=u_1u_2+u_2\\
\alpha_3=\alpha_2-u_2=u_1u_2\\
\alpha_4=u_3u_4+u_3+u_4\\
\alpha_5=\alpha_4-u_3=u_3u_4+u_4\\
\alpha_3=\alpha_5-u_4=u_3u_4\\
\end{array}\]

So $u_1u_2=\alpha_3=u_3u_4$. The above also gives a satisfying assignment for the auxiliary variables when $u_1u_2=u_3u_4$. Note that if the $u_i$ each have absolute value at most $a$, then this assignment of the auxiliary variables has $|\alpha_i|\le 2a+a^2$. 

For each constraint on $V$, we can now specify corresponding constraints on $V'$. Using in some cases the $u_1u_2=u_3u_4$ constraint as a building block, the following table shows how to do this:

\begin{center}
\begin{tabular}{|c|c|c|c|}
\hline
Constraint type&Aux. Vars&Constraints on $V'$&Satisfying assignment\\
\hline
$x=y$&$\alpha_1$&$\begin{array}{c}
x'+\alpha_1=y'\\y'+\alpha_1=x'
\end{array}$&$\alpha_1=0$\\
\hline
$x+y=z$&none&$x'+y'=z'$&-\\
\hline
$x+y=z+w$&$\alpha_1$&$\begin{array}{c}
x'+y'=\alpha_1\\z'+w'=\alpha_1
\end{array}$&$\alpha_1=x'+y'$\\
\hline
$x=(1+q)y$&$\alpha_1$&$\begin{array}{c}
(1+\alpha_q)(1+y')=(1+\alpha_1)\\x'+\alpha_q=\alpha_1
\end{array}$
&$\begin{array}{c}\alpha_1=x'+q\end{array}$\\
\hline
$xy=z$&$\begin{array}{c}5\end{array}$&$\begin{array}{c}
x'y'=z'\alpha_q
\end{array}$&-\\
\hline
$xy=zw$&$\begin{array}{c}5\end{array}$&$x'y'=z'w'$&-\\
\hline
\end{tabular}
\end{center}

Note that a constraint $xy=z$ on $V$ should correspond to a constraint equivalent to $qx'\cdot qy'=qz'$ on $V'$. 

For each constraint type, it is straightforward to verify that if a point $(x', y)\in \mathbb{R}^{74}$ satisfies the constraints in the third column, then $q^{-1}x'$ satisfies the original constraint. 

Also, for each constraint type, if $x\in \mathbb{R}^4$ satisfies the constraint, then we can assign values in $[-6q, 6q]$ to the auxiliary variables (include the extra auxiliary variables from the $u_1u_2=u_3u_4$ constraints) so that the constraints in the third column are satisfied by $(qx, y)\in \mathbb{R}^{74}$ for some $y\in \mathbb{R}^{70}$. 

Since $V$ is a subset of $[\frac12, 2]^4$, each variable in $V$ also implicitly has constraints $x\ge \frac12$ and $x\le 2$. We so need to represent these range constraints:


\begin{center}
\begin{tabular}{|c|c|c|c|}
\hline
Bound type&Aux. Vars&Constraints on $V'$&Satisfying assignment\\
\hline
$x\le 2$&$\begin{array}{c}
\alpha_1,\alpha_2,\alpha_3\\
+5
\end{array}$&$\begin{array}{c}
\alpha_1+\alpha_1=x'\\\alpha_1+\alpha_2=\alpha_q\\
q\alpha_2=\alpha_3\alpha_3\\
\end{array}$&$\begin{array}{c}
\alpha_1=\frac12x'\\
\alpha_2=q-\frac12x'\\
\alpha_3=\sqrt{q\alpha_2}
\end{array}$\\
\hline
$x\ge \frac12$&$\begin{array}{c}
\alpha_1,\alpha_2,\alpha_3\\
+5
\end{array}$&$\begin{array}{c}
x'+x'=\alpha_1\\\alpha_2+\alpha_q=\alpha_1\\
q\alpha_2=\alpha_3\alpha_3\\
\end{array}$&$\begin{array}{c}
\alpha_1=2x'\\
\alpha_2=2x'-q\\
\alpha_3=\sqrt{q\alpha_2}
\end{array}$\\
\hline
\end{tabular}
\end{center}

In each case, the idea is that we can only assign a real value $\alpha_3=\sqrt{q\alpha_2}$ if $\alpha_2$ is positive. A satisfying assignment has $\alpha_2\in [0, 4q]$ (at worst), so $\alpha_3=\sqrt{q\alpha_2}\in [0, 2q]$. So the extra auxiliary variables needed for the $u_1u_2=u_3u_4$ constraint have absolute value at most $4q+(2q)^2\le 6q$.

Each variable in $V$ gets one upper bound constraint $x\le 2$ and one lower bound constraint $x\ge \frac12$. 

Each constraint type requires $8$ auxiliary variables for each of the upper and lower bounds on each of the up to $4$ input variables, up to $5$ auxiliary variables for the constraint and the one auxiliary variable $\alpha_q$, giving a total of $70$ auxiliary variables.

There is at most one constraint of form $x=q$. Each upper or lower bound constraints uses $6$ addition constraints and $2$ multiplication constraints. Each of the constraints in the first table uses at most $4$ addition and $2$ multiplication constraints. So in total we need at most $52$ constraints of form $x+y=z$ and $18$ constraints of form $(1+x)(1+y)=(1+z)$.  
\end{proof}

We can now prove \Cref{lem:threeconstraintconversion}. Recall that we are given $n\in \mathbb{N}$, and we want to find an $m\in \mathbb{M}$ such that, for any $q\in (0, \frac12]$ and a constraint $V\in \mathcal{P}\left(\mathcal{C}(q)_{[\frac12, 2]}^{n}\right)$, we can construct a subset $V'$ of $\mathbb{R}^{n+m}$ satisfying:

\begin{itemize}
    \item $V'$ is cut out by constraints of form $x+y=z$, $(1+x)(1+y)=(1+z)$, and $x=q$, with exactly one constraint of form $x=q$,
    \item for any $x\in V$, there is a $y\in [-6q, 6q]^m$ such that $(qx, y)\in V'$, and
    \item if $(x, y)\in V'$ for some $x\in \mathbb{R}^n$, then $q^{-1}x\in V$
\end{itemize}

\begin{proof}[Proof of \Cref{lem:threeconstraintconversion}]
$V$ is cut out by a conjunction of constraints of the form in $\mathcal{C}(q)$. For each of these constraints, we use the first table from the proof of \Cref{lem:constraintshrinking} to create a corresponding set of constraints on $V'$. For each variable, we use the constraints specified by the second table in the proof of \Cref{lem:constraintshrinking} to add an upper and lower bound. 

This might create many variables with constraints $x=q$. However, these can be merged into one auxiliary variable $\alpha_q$.

The constraints $x=y$, $x=(1+q)y$, $x+y=z$ and $x+y=z+w$ are linear, so we can assume that $V$ involves at most $n$ of these constraints. If there are more, then some of the constraints are redundant. Each of these constraints requires at most $1$ auxiliary variable. 

Similarly, the constraints $xy=z$ and $xy=zw$ become linear after taking the logarithm (note that the variables in appearing in such constraints are restrict to being positive). So we can similarly assume that there are at most $n$ of these constraints, each requiring $5$ auxiliary variables. 

Each variable has an upper and a lower bound constraint, which each require $8$ auxiliary variables. So in total, we need at most $m=1+22n$ auxiliary variables.
\end{proof}

\paragraph{A gap-preserving reduction from MAX-ETR-$\mathcal{C}(q)$ to MAX-ETR-INV.}

Recall that an instance of MAX-ETR-INV asks us to find an assignment of variables in $[\frac12, 2]$ that minimizes the fraction of unsatisfied constraints, where the constraints are of form $x=1$, $xy=1$ or $x+y=z$. 

\Cref{thm:etrq} says that Gap-MAX-ETR-$\mathcal{C}(q)_c^1$ is $\exists\mathbb{R}$ hard for sufficiently small $q$ and some $c<1$. We now want to transform an instance of Gap-MAX-ETR-$\mathcal{C}(q)_c^1$ into an instance of Gap-MAX-ETR-INV$_{c'}^1$. 

Given an instance $\varphi$ of MAX-ETR-INV or MAX-ETR-$\mathcal{C}(q)$, we write $\text{UNSAT}(\varphi)$ for the minimum fraction of unsatisfied constraints. So the optimal objective value is $1-\text{UNSAT}(\varphi)$. We can now state the gap-preserving reduction:

\begin{lemma}\label{lem:gapetrinv}
Let $q=2^{-k}$ for some integer $k\ge 6$. Then there is a polynomial time algorithm that, given an instance $\varphi$ of MAX-ETR-$\mathcal{C}(q)$, produces an instance $\psi$ of MAX-ETR-INV such that $(1266+2k)^{-1}\cdot\text{UNSAT}(\varphi)\le \text{UNSAT}(\psi)\le \text{UNSAT}(\varphi)$
\end{lemma}

\begin{proof}
For each variable $x$ in $\varphi$, we create a corresponding variable $x''$ in $\psi$, where $x''$ should represent the value of $1+xq$. 

For every constraint in $\varphi$, we now want to create system of constraints on $\psi$ involving some auxiliary variables.

We first use \Cref{lem:constraintshrinking} to replace each constraint in $\varphi$ with a system of constraints of form $x'=q$, $x'+y'=z'$ or $(1+x')(1+y')=(1+z')$, where now the variables are in $[-6q, 6q]\subseteq [-\frac18, \frac18]$. 

For each variable $x'$ in the new system, we create a variable $x'$ in $\psi$, which represents the value of $1+x\in [\frac78, \frac98]$. 

Next, we need to create some variables in $\psi$ representing constants $1, \frac12, \frac32, \frac34, \frac74$ and $\frac78$:

\[V_1=1\]
\[V_{1/2}+V_{1/2}=V_1\]
\[V_{1/2}+V_1=V_{3/2}\]
\[V_{3/4}+V_{3/4}=V_{3/2}\]
\[V_{3/4}+V_1=V_{7/4}\]
\[V_{7/8}+V_{7/8}=V_{7/4}\]

Each constraint in $\varphi$ gets its own set of constant variables and constraints setting the constant values. This is necessary to ensure enough redundancy in these constraints. 

A constraint $x+y=z$ corresponds to $x'+y'=z'+1$. For such a constraint, we create auxiliary variables $\alpha_1,\dots,\alpha_{7}$ and create the following $8$ constraints on $\psi$:

\begin{align*}
x'+V_{\frac12}=\alpha_1&\quad(\alpha_1=x'+\tfrac12\in[\tfrac{11}8, \tfrac{13}8])\\
y'+V_{\frac12}=\alpha_2&\quad(\alpha_2=y'+\tfrac12\in[\tfrac{11}8, \tfrac{13}8])\\
z'+V_{\frac12}=\alpha_3&\quad(\alpha_3=z'+\tfrac12\in[\tfrac{11}8, \tfrac{13}8])\\
\alpha_4+\alpha_4=\alpha_1&\quad(\alpha_4=\tfrac12x'+\tfrac14\in [\tfrac{11}{16}, \tfrac{13}{16}])\\
\alpha_5+\alpha_5=\alpha_2&\quad(\alpha_5=\tfrac12y'+\tfrac14\in [\tfrac{11}{16}, \tfrac{13}{16}])\\
\alpha_6+\alpha_6=\alpha_3&\quad(\alpha_6=\tfrac12z'+\tfrac14\in [\tfrac{11}{16}, \tfrac{13}{16}])\\
\alpha_6+V_{\frac34}=\alpha_{7}&\quad(\alpha_{7}=\tfrac12z'+1\in [\tfrac{27}{16}, \tfrac{29}{16}])\\
\alpha_4+\alpha_5=\alpha_7&\quad(\tfrac12x'+\tfrac12y'+\tfrac12=\tfrac12z'+1,x'+y'=z'+1)\\
\end{align*}

The multiplication constraint is harder. Luckily, most of the work has already been done by Abrahamsen, Adamaszek, and Miltzow (\cite{AAM2021}, Section 4) in their proof that ETR-INV is $\exists\mathbb{R}$-hard.

In their proof, they show that there is a system of constraints of form $xy=1$ and $x+y=z$ involving input variables $x, y$ and $z$ and some auxiliary variables such that:

\begin{itemize}
    \item if all the constraints are satisfied, then $xy=z$, and
    \item if $x,y\in [\frac78, \frac98]$ and $z=xy$, then there is an assignment of auxiliary variables, all in $[\frac12, 2]$, such that the constraints are satisfied
\end{itemize}

By our count, their construction has $47$ constraints and $46$ auxiliary variables (not counting the constraints used to create constant values $1, \frac12, \frac32, \frac34, \frac74$ and $\frac78$). 


We may need a constraint of form $x'=q$. This should correspond to $x''=1+q=1+2^{-k}$. We create auxiliary variables $\alpha_i$ for $i=2, \dots, k-1$ and $\beta_i$ for $i=1,\dots,k-1$. 

For $i=1, \dots, k-1$, we create constraints $\beta_{i}+\beta_{i}=\alpha_i$ and $\beta_i+V_{\frac12}=\alpha_{i+1}$, where $\alpha_1$ is an alternate label for $V_{\frac32}$ and $\alpha_{k}$ is an alternate label for $x''$.

Inductively, $\alpha_i=1+2^{-i}$ requires $\beta_i=\frac12+2^{-i-1}$, so $\alpha_{i+1}=1+2^{-i-1}$. So if all these constraints are satisfied (including the ones creating the constants $V_{\frac12}$ and $V_{\frac32}$), then we have $x''=1+q$.

\Cref{lem:constraintshrinking} writes any of the constraint types in $\varphi$ in terms of at most $1$ constraint of form $x'=q$, $52$ constraints of form $x'+y'=z'$, and $18$ constraints of form $(1+x')(1+y')=(1+z')$. 

So in total, each constraint from $\varphi$ requires $6$ constraints to create constants, up to $8\cdot 52$ constraints representing constraints of form $x'+y'=z'$, up to $18\cdot 47$ constraints representing constraints of form $(1+x')(1+y')=(1+z')$, and up to $2k-2$ constraints to create a constraint $x'=q$. In total, this gives $1266+2k$ constraints on $\psi$ for every constraint in $\varphi$. If we have less than this many constraints, then we can add repeats to get exactly $1266+2k$.

Clearly, $\psi$ is satisfiable if and only if $\varphi$ is satisfiable. Each constraint in $\varphi$ is represented by a set of $1266+2k$ constraints in $\psi$ with some auxiliary variables, so that $(x, y, z, w)\in \mathbb{R}^4$ satisfy that constraint in $\varphi$ if an only if there is an assignment of the auxiliary variables so that $(x'', y'',z'',w'')=(1+qx, 1+qy, 1+qz, 1+qw)$ satisfies the $1266+2k$ corresponding constraints. 

Suppose there is an assignment of variables in $\varphi$ that satisfies a $1-\epsilon$ fraction of constraints. Then we can assign the variables in $\psi$ by $x''=1+qx$, and for each constraint that is satisfied in $\varphi$ we can set the auxiliary variables corresponding to that constraint so that the corresponding constraints in $\psi$ are satisfied. Since each constraint in $\varphi$ corresponds to exactly $1266+2k$ constraints in $\psi$, this assignment satisfies at least a $1-\epsilon$ fraction of constraints in $\psi$. So $\text{UNSAT}(\psi)\le \text{UNSAT}(\varphi)$.

Conversely, suppose that there is an assignment of variables in $\psi$ that satisfies a $1-\epsilon$ fraction of constraints in $\psi$. Then we can assign variables in $\varphi$ by $x=q^{-1}(x''-1)$, or $x=1$ if $q^{-1}(x''-1)\notin [-2, 2]$. 

The constraints in $\psi$ are grouped into sets of size $1266+2k$, one set for each constraint in $\varphi$. Out of these sets, at least a $1-(1266+2k)\epsilon$ fraction of them have all their constraints satisfied. The corresponding constraints in $\varphi$ are then satisfied by the assignment (note that the constraints generated by \Cref{lem:constraintshrinking} require the variables to be in the appropriate range, if they are satisfied). So $(1266+2k)^{-1}\cdot\text{UNSAT}(\varphi)\le \text{UNSAT}(\psi)$.
\end{proof}

Now \Cref{thm:etrq} and \Cref{lem:cqhardness,lem:gapetrinv} imply that there is some $c<1$ such that Gap-MAX-ETR-INV$_{c}^1$ is $\exists\mathbb{R}$-hard, showing that MAX-ETR-INV is $\exists\mathbb{R}$-hard to approximate. It remains to prove \Cref{lem:main}. 

\section{Complexity of solutions}\label{sec:solutioncomplexity}

In this section, we prove \Cref{thm:solutioncomplexity}. We are given an integer $k$, and we want to produce a satisfiable instance of ETR-INV such that any assignment satisfying a $c$ fraction of constraints must have at least one variable set to a value that generates a field extension of degree more than $n^k$, where $n\ge 2$ is the number of variables in the instance. We let $c$ be the constant from \Cref{thm:main}. 

For each integer $d$, let $\mathcal{A}(d)$ be the field generated by real algebraic numbers that are roots of integer-coefficient polynomials of degree at most $d$. Given an instance $\varphi$ of ETR-INV or ETR-$\mathcal{C}(q)$, we let $\text{UNSAT}_d(\varphi)$ be the minimum over variable assignments in $\mathcal{A}(d)\cap [\frac12, 2]$ of the number of unsatisfied constraints.

The idea is to start with an instance $\varphi$ of ETR-$\mathcal{C}(q)$ that is satisfiable, but not satisfiable over $\mathcal{A}(d)$ for some $d$. So $\text{UNSAT}(\varphi)=0$ but $\text{UNSAT}_d(\varphi)>0$. We then apply the gap amplification process described in the proof \Cref{thm:etrq} to $\varphi$, obtaining a new instance $\psi$ of ETR-$\mathcal{C}(q)$. We know that $\text{UNSAT}(\psi)=0$, and we claim that $\text{UNSAT}_{d}(\psi)\ge 1-c$.

The transformations given by \Cref{lem:step1,lem:step2} don't depend on the underlying alphabet, so we can directly apply these steps over the alphabet $\mathcal{A}(d)\cap [\frac12, 2]$ rather than $[\frac12, 2]$. We need to be more careful about \Cref{cor:step3}. 

\begin{lemma}
Let $G_2$ be a constraint graph with constraint set $\mathcal{P}\left(\mathcal{C}(q)_{[\frac12, 2]}^{2n}\right)$ and let $G_3$ be the constraint graph produced by \Cref{cor:step3}. Then for any integer $d$, $\text{UNSAT}_{d}(G_{3})\ge \eta_3\cdot \text{UNSAT}_b(G_2)$ where $\eta_3=\frac1{5000}$.
\end{lemma}

\begin{proof}
Each vertex in $G_2$ is replaced by $2^n$ vertices in $G_3$ corresponding to the values of the error-correcting function $f$. Given an assignment of the variables in $G_3$ that satisfies a $1-\epsilon$ fraction of edges in $G_3$, \Cref{lem:composition} produces an assignment of the variables in $G_2$ that satisfies a $1-5000\epsilon$ fraction of edges in $G_2$. If the assignment of $G_3$ has values in $\mathcal{A}(d)$, then we need to check that the assignment of $G_2$ produced by the lemma also has values in $\mathcal{A}(d)$. 

The error-correcting function $f$ maps $(\alpha_1, \dots, \alpha_n)$ to $(a_\sigma)_{\sigma\in \{0, 1\}^n}$ where $a_\sigma=\prod_{i=1}^n(1+q\alpha_i)^{\sigma_i}$. Given an assignment of variables in $G_3$, the assignment produced by \Cref{lem:composition} assigns each vertex of $G_2$ the value that with encoding closest to the value of the corresponding variables in $G_3$, with ties broken arbitrarily. Examining the proof of \Cref{lem:composition}, we notice that it is only necessary to choose this value when the distance is less than half the distance of the code (so less than $\frac{1}{4}$). If the distance is $\ge \frac14$, then we can instead assign a value to that vertex arbitrarily. 

So it is sufficient to show that, if $(a_\sigma)_{\sigma\in \{0, 1\}^n}$ are in $\mathcal{A}(d)$ and $a_\sigma=\prod_{i=1}^n(1+q\alpha_i)^{\sigma_i}$ for more than $\frac34$ of $\sigma$, then the $\alpha_i$ are in $\mathcal{A}(d)$.

Let $S$ be the set of $\sigma$ such that $a_\sigma=\prod_{i=1}^n(1+q\alpha_i)^{\sigma_i}$. For each $i$, there is some $\sigma$ with $\sigma(i)=0$ such that both $\sigma$ and $\sigma+e_i$ are in $S$, so $\sigma_i=\frac{a_{\sigma+e_i}}{a_\sigma}$, which is in $\mathcal{A}(d)$.
\end{proof}

Next, we need to find instances of ETR-$\mathcal{C}(q)$ that are satisfiable, but not over $\mathcal{A}(d)$ for some $d$. First, it will be helpful to find some real algebraic numbers that aren't in $\mathcal{A}(d)$.

\begin{lemma}
For every $\alpha\in \mathcal{A}(d)$, the prime factors of $[\mathbb{Q}(\alpha):\mathbb{Q}]$ (the degree of the field extension $\mathbb{Q}(\alpha)$ over $\mathbb{Q}$) are at most $d$. 
\end{lemma}

\begin{proof}
The field $\mathcal{A}(d)$ is generated by roots of integer-coefficient polynomials of degree at most $d$. For $\alpha\in \mathcal{A}(d)$, there is a finite set of generators $\alpha_1, \dots, \alpha_j$ such that $\alpha\in \mathbb{Q}(\alpha_1, \dots, \alpha_j)$. By the tower law for field extensions, $[\mathbb{Q}(\alpha_1, \dots, \alpha_j):\mathbb{Q}]$ can be written as a product:

\[[\mathbb{Q}(\alpha_1, \dots, \alpha_j):\mathbb{Q}]=[\mathbb{Q}(\alpha_1, \dots, \alpha_j):\mathbb{Q}(\alpha_1, \dots, \alpha_{j-1})]\cdots[\mathbb{Q}(\alpha_1):\mathbb{Q}]\]

Each $\alpha_i$ is a root of a polynomial of degree at most $d$ over $\mathbb{Q}$, so is a root of a polynomial of degree at most $d$ over $\mathbb{Q}(\alpha_1, \dots, \alpha_{i-1})$. So each $[\mathbb{Q}(\alpha_1, \dots, \alpha_i):\mathbb{Q}(\alpha_1, \dots, \alpha_{i-1})]$ is at most $d$. So the prime factorization of $\mathbb{Q}(\alpha_1, \dots, \alpha_j)$ contains no factors factors larger than $d$. 

Again by the tower law, we have:

\[[\mathbb{Q}(\alpha_1, \dots, \alpha_j):\mathbb{Q}]=[\mathbb{Q}(\alpha_1, \dots, \alpha_j):\mathbb{Q}(\alpha)]\cdot [\mathbb{Q}(\alpha):\mathbb{Q}]\]

So $[\mathbb{Q}(\alpha):\mathbb{Q}]$ divides $\mathbb{Q}(\alpha_1, \dots, \alpha_j)$, and therefore also has no prime factors larger than $d$. 
\end{proof}

For an integer $p$, the degree of the field extension $\mathbb{Q}\left(2^{\frac1p}\right)/\mathbb{Q}$ is $p$. So if $p$ is prime, then $2^{\frac1p}$ is not in $\mathcal{A}(d)$ for any $d<p$. 

\begin{lemma}\label{lem:cqbadsolution}
For any $q>0$ and integer $k$, there is an instance of ETR-$\mathcal{C}(q)$ that is satisfiable, but not satisfiable over $\mathcal{A}(n^k)$, where $n$ is the number of variables in the instance. 
\end{lemma}

\begin{proof}
For a prime $p$, we want to construct an instance of ETR-$\mathcal{C}(q)$ where a solution requires one variable $x_0$ to satisfy $x_0^p=2$.

Let $j=\lceil \log_2(p)\rceil$, and add variables $x_1, \dots, x_j$. We add a constraint $x_i\cdot x_i=x_{i+1}$ for each $i\in \{0, \dots, j-1\}$. Inductively, a satisfying assignment must have $x_i=x_0^{2^i}$ for each $i\in \{0, \dots, j\}$. 

Now write $p$ as a binary expansion $p=\sum_{i=0}^{j}b_i2^i$ with $b_i\in \{0, 1\}$. Add variables $y_0, \dots, y_j$. If $b_0=1$ then add a constraint $y_0=x_0$, otherwise add a constraint $y_0\cdot y_0=y_0$ (requiring $y_0=1$). For each $i\in \{1, \dots, j\}$, we create a constraint $y_{i}=y_{i-1}$ if $b_i=0$, or $y_{i}=x_{i}y_{i-1}$ if $b_i=1$. A satisfying assignment must have $y_j=x_0^p$. 

We now add a variable $V_1$ and constraints $V_1\cdot V_1=V_1$ and $V_1+V_1=y_j$, requiring that $y_j=2$, and so $x_0=2^{\frac1p}$. It is straightforward to check that there is a satisfying assignment where each $x_i$ and $y_i$ is of form $2^{\frac{\ell}{p}}$ for some $\ell\le p$, so is in $[\frac12, 2]$.

This construction requires $2\lceil \log_2(p)\rceil+3$ variables, so choosing $p$ large enough that $p>\left(2\lceil \log_2(p)\rceil+3\right)^k$ establishes the claim.
\end{proof}

Using the gap amplification process on an instance generated by \Cref{lem:cqbadsolution}, we get an instance of ETR-$\mathcal{C}(q)$ that is satisfiable over $[\frac12, 2]$ but with constant UNSAT value over $[\frac12, 2]\cap \mathcal{A}(d)$. To get an instance of ETR-INV, we just need to check that \Cref{lem:gapetrinv} preserves gaps over $[\frac12, 2]\cap \mathcal{A}(d)$.

\begin{lemma}
Let $q$ be a (negative) power of two, let $\varphi$ be an instance of MAX-ETR-$\mathcal{C}(q)$ and let $\psi$ be the instance of MAX-ETR-INV produced by \Cref{lem:gapetrinv}. Then $(1266-\log_2(q))^{-1}\cdot \text{UNSAT}(\varphi)\le \text{UNSAT}(\psi)$.
\end{lemma}

\begin{proof}
Given an assignment of the variables in $\psi$ that satisfies a $1-\epsilon$ fraction of constraints in $\psi$, the proof of \Cref{lem:gapetrinv} produces an assignment of variables in $\varphi$ that satisfies at least a $1-(1266-\log_2(q))\epsilon$ fraction of constraints in $\psi$. Each variable $x''$ in $\varphi$ is assigned a value of $1+qx$ for some corresponding variable $x$ in $\psi$. Since $1$ and $q$ are rational, $1+qx$ is in $\mathcal{A}(d)$ if $x$ is. 
\end{proof}

To prove \Cref{thm:solutioncomplexity}, we start with an instance of ETR-$\mathcal{C}(q)$ produced by \Cref{lem:cqbadsolution}, apply the gap amplification process and then \Cref{lem:gapetrinv}. We obtain an instance $\psi$ of ETR-INV that is satisfiable, but $\text{UNSAT}_d(\psi)\ge 1-c$. The gap amplification process and \Cref{lem:gapetrinv} each introduce only a polynomial increase in the number of variables, so using \Cref{lem:cqbadsolution} we can choose $d$ large enough that the claim of \Cref{thm:solutioncomplexity} is satisfied.

\section{Approximating Max-ETR-INV}\label{sec:algorithms}

We now give a polynomial-time $8$-factor approximation algorithm for Max-ETR-INV.

\begin{theorem}
There is a polynomial-time algorithm that decides Gap-Max-ETR-INV$_a^b$ for any $0\le a<b\le 1$ with $b/a>8$.
\end{theorem}

\begin{proof}
For each variable $x$, we assign it a value of either $1$ or $2$ with equal probability. Now each constraint of form $x=1$ is satisfied with probability $\frac12$.

Each constraint of form $xy=1$ is satisfied with probability $\frac14$ if $x$ and $y$ are different variables, or $\frac12$ if $x$ and $y$ are the same variable. Each constraint of form $x+y=z$ is satisfied with probability $\frac18$ if $x, y$ and $z$ are all different variables, or $\frac14$ if $x$ and $y$ are the same variable but $z$ is different.

Note that constraints of form $x+y=x$ or $x+x=x$ are not possible to satisfy with $x,y\in [\frac12, 2]$. Let $A$ be the set of constraints that are possible to satisfy (that is, those that aren't of form $x+y=x$ or $x+x=x$). Then every constraint in $A$ is satisfied with probability at least $\frac18$.

So the expected number of satisfied clauses is at least $\frac18|A|$, so there is some assignment satisfying this many clauses. We know that an optimal assignment satisfies at most $|A|$ clauses, so if $|A|< b$ then we can answer ``no'', otherwise $\frac18|A|>a$ and we can answer ``yes''.
\end{proof}

Note that this algorithm only approximates the size of an optimal solution, we need randomness in order to find an good assignment.

\Cref{thm:main} says that even a non-deterministic algorithm can't decide Gap-Max-ETR-INV$_{c}^1$ for some $c<1$, unless NP$=\exists\mathbb{R}$. However, with the additional power of NP, we can achieve a better approximation factor.

\begin{theorem}
There is a non-deterministic-polynomial-time algorithm that decides Gap-Max-ETR-INV$_a^b$ for any $0\le a<b\le 1$ with $b/a> 2$.
\end{theorem}

\begin{proof}
Let $\varphi$ be an instance of Gap-Max-ETR-INV$_a^b$ with $b/a>2$. If it is possible to satisfy some $\ge b$ of the constraints in $\varphi$, then we want to give a certificate that it is possible to satisfy some $>a$ fraction of constraints in $\varphi$. 

We split the constraints in $\varphi$ into two subsystems. Let $A$ be the subsystem of constraints of form $xy=1$ and $x=1$ and let $B$ be the subsystem of constraints of form $x+y=z$. 

Note that we can satisfy all the constraints in $A$ by setting all the variables equal to $1$. If $\frac{|A|}{|B|}>a$, then we are done.

Let $\beta$ be the maximum fraction of satisfied constraints in $B$ over all assignments of the variables. Now any assignment of variables satisfies at most a $\frac{|A|+\beta|B|}{|A|+|B|}$ fraction of constraints in $\varphi$. So if it is possible to satisfy a $\ge b$ fraction of constraints in $\varphi$ and $\frac{|A|}{|B|}\le a$, then $\frac{\beta|B|}{|A|+|B|}\ge b-a>a$.

In this case, the certificate contains an assignment of variables that satisfies a $\beta$ fraction of constraints in $B$, represented as rational numbers with polynomially many bits. Since the constraints in $B$ (and the implicit constraints $x\in [\frac12, 2]$) are affine, such an assignment exists. This assignment satisfies more than an $a$ fraction of constraints in $\varphi$. 
\end{proof}

\section{Acknowledgments}

I am grateful to anonymous reviewers for feedback on earlier drafts of this paper. This work is supported by Independent Research Fund Denmark, grant 1054-00032B, and by the Carlsberg Foundation, grant CF24-1929.

\bibliographystyle{alphaurl}
\bibliography{bib}

\end{document}